\documentclass[11pt]{article}
\usepackage{amsmath,amsthm,amssymb,tikz}
\usetikzlibrary{calc}
\usepackage{hyperref}
\usepackage{ulem}
\usepackage{fancyvrb}
\hypersetup{
pdfauthor = {Boris Bukh, Po-Shen Loh, Gabriel Nivasch},
pdfsubject = {Mathematics},
pdfkeywords = {Ramsey theorem, Erdos-Szekeres theorem, Radon point, Tverberg partition},
pdfstartview = {FitH},
pdfpagemode = {UseNone}}

\hypersetup{pdftitle = {Classifying unavoidable Tverberg partitions}}

\title{Classifying unavoidable Tverberg partitions}
\author{Boris Bukh\thanks{Department of Mathematical Sciences, Carnegie
Mellon University, Pittsburgh, PA 15213, USA. Research was supported in
part by Churchill College, Cambridge, by U.S. taxpayers through NSF grant
DMS-1301548, and by Alfred P.\ Sloan Foundation through Sloan Research
Fellowship.} \and Po-Shen Loh\thanks{Department of Mathematical Sciences,
  Carnegie Mellon University, Pittsburgh, PA 15213, USA. Research supported
  in part by an NSA Young Investigators Grant, NSF grants 
DMS-1201380 and DMS-1455125, and by a USA-Israel BSF Grant.} \and Gabriel Nivasch\thanks{Department of Computer Science, Ariel University, Ariel, Israel. Research was partially supported by ERC Advanced Research Grant 267165 (DISCONV).}}
\date{}

\theoremstyle{plain}
\newtheorem{theorem}{Theorem}[section]
\newtheorem{lemma}[theorem]{Lemma}
\newtheorem{observation}[theorem]{Observation}

\newtheorem{corollary}[theorem]{Corollary}
\newtheorem{conjecture}[theorem]{Conjecture}                  
\newtheorem{claim}{Claim}                             

\theoremstyle{definition}
\newtheorem{definition}[theorem]{Definition}

\theoremstyle{remark}
\newtheorem{remark}[theorem]{Remark}

\newcommand*{\R}{\mathbb{R}}                          

\newcommand*{\eqdef}{\stackrel{\text{\tiny{def}}}{=}} 
\def\mygobble#1{}
\def\vepsapproximant/{$\varepsilon$\nobreakdash-approx\-i\-mant}
\def\vepsapproximants/{$\varepsilon$\nobreakdash-approx\-i\-mants}
\def\vepsnet/{$\varepsilon$\nobreakdash-net}
\def\vepsnets/{$\varepsilon$\nobreakdash-nets}

\newcommand*{\TODO}[1]{}                              
\newcommand*{\NB}[1]{}                                

\DeclareMathOperator{\conv}{conv}                     
\DeclareMathOperator{\ahull}{ahull}                   
\DeclareMathOperator{\orient}{orient}                 
\DeclareMathOperator{\tv}{Tv}                         
\DeclareMathOperator{\sixpt}{\Pi_{\mathrm{six}}}              

\newcommand{\Gs}{G_{\mathrm s}}                       
\newcommand{\B}{\mathcal{B}}                          
\DeclareMathOperator{\sconv}{\mathrm{stconv}}         

\newcommand{\codefont}[1]{{\small\texttt{#1}}}

\begin{document}

\maketitle

\begin{abstract}
Let $T(d,r) \eqdef (r-1)(d+1)+1$ be the parameter in Tverberg's theorem, and call a partition $\mathcal I$ of $\{1,2,\ldots,T(d,r)\}$ into $r$ parts a \emph{Tverberg type}. We say that $\mathcal I$ \emph{occurs} in an ordered point sequence $P$ if $P$ contains a subsequence $P'$ of $T(d,r)$ points such that the partition of $P'$ that is order-isomorphic to $\mathcal I$ is a Tverberg partition. We say that $\mathcal I$ is \emph{unavoidable} if it occurs in every sufficiently long point sequence.

In this paper we study the problem of determining which Tverberg types are unavoidable. We conjecture a complete characterization of the unavoidable Tverberg types, and we prove some cases of our conjecture for $d\le 4$.
Along the way, we study the avoidability of many other geometric predicates.

Our techniques also yield a large family of $T(d,r)$-point sets for which the number of Tverberg partitions is exactly $(r-1)!^d$. This lends further support for Sierksma's conjecture on the number of Tverberg partitions.

{\bf Keywords:} Geometric predicate, Ramsey theory, stair-convexity, Tverberg's theorem.
\end{abstract}

\section{Introduction}

By a \emph{(geometric) predicate} of arity $k$ we mean a $k$-ary relation $\Phi$ on $\R^d$, i.e. a property which $k$-tuples of points $(p_1, \ldots, p_k)\in(\R^d)^k$ might or might not satisfy. We focus on \emph{semialgebraic} predicates, which are predicates given by Boolean combinations of terms of the form $f(p_1,\ldots,p_k)>0$, where the $f$'s are nonzero polynomials.

An example of a geometric predicate is the planar predicate ``$p_1,p_2,\allowbreak p_3,p_4\in \R^2$ are in convex position''. (We show in Section~\ref{sec_semi} that all predicates we consider are semialgebraic.) Another example is the $(d+1)$-ary \emph{orientation} predicate in $\R^d$, $\orient(p_1,\ldots, p_{d+1})$, whose defining polynomial is
\begin{equation*}
\det\left[
\begin{array}{ccc}
1&\cdots&1\\
p_1&\cdots&p_{d+1}
\end{array}
\right].
\end{equation*}
For example, in the plane, $\orient(p_1,p_2,p_3)$ means that $p_1,p_2,p_3$ are in counterclockwise order.

Let $P\eqdef (p_1, \ldots, p_n)\in(\R^d)^n$ be an ordered sequence of $n$ points in $\R^d$. We say that a $k$-ary predicate $\Phi$ \emph{occurs} in $P$ if $P$ contains a subsequence $p_{i_1}, \ldots, p_{i_k}$, $i_1<\cdots<i_k$, for which $\Phi(p_{i_1}, \ldots, p_{i_k})$ holds. Otherwise, we say that $P$ \emph{avoids} $\Phi$. For simplicity, we will assume that $P$ is \emph{$\Phi$-generic}, which means that none of the polynomials defining $\Phi$ evaluate to $0$ on any $k$-tuple of distinct points of $P$. Genericity can be achieved, if necessary, by an appropriate arbitrarily small perturbation of $P$.

We say that $P$ is \emph{homogeneous} with respect to $\Phi$ if $P$ avoids either $\Phi$ or $\neg \Phi$ (the negation of $\Phi$). Ramsey's theorem implies the existence of arbitrarily-long homogeneous point sequences with respect to any predicate: Indeed, given $\ell$, Ramsey's theorem states that there exists a large enough $n$ such that \emph{every} $n$-point sequence $P$ contains a $\Phi$-homogeneous subsequence of length $\ell$.

Point sequences that are homogeneous with respect to the orientation predicate are important because they form the vertices of cyclic polytopes~(see e.g.~\cite{ziegler}); these point sequences have been studied previously e.g.~in \cite{BN_onesided,BMP,EMRPS,suk}, and specifically in the context of Tverberg's theorem in \cite{tolerated_tv,PS}. In the plane, $P$ is orientation-homogeneous if and only if its points are in convex position and are listed in the order they appear along the boundary of $\conv P$.

We say that a predicate $\Phi$ is \emph{unavoidable} if it occurs in every sufficiently-long $\Phi$-generic point sequence. Otherwise, if there exist arbitrarily long point sequences avoiding $\Phi$, then $\Phi$ is \emph{avoidable}.

For example, the planar $4$-points-in-convex-position predicate mentioned above is unavoidable, since every generic five-point set contains four points in convex position. This result is known as the ``happy ending problem'', and it is actually one of the original results that motivated the development of Ramsey theory~\cite{ramsey_erdos}.

Hence, for each predicate $\Phi$ there are three mutually exclusive possibilities: Either $\Phi$ is unavoidable, or $\neg \Phi$ (the negation of $\Phi$) is unavoidable, or both $\Phi$ and $\neg \Phi$ are avoidable. For us in this paper, to \emph{solve} a predicate means to determine on which of these categories it falls.

This problem has been shown to be decidable for semialgebraic predicates in dimension $1$~\cite{BM_decidability}. However, for dimensions $d\ge 2$ the problem remains open.

In this paper we focus on predicates related to Tverberg partitions. We present a conjecture and prove some partial results in low dimensions. We also examine along the way some related predicates. 

\subsection{Tverberg partitions}

For a positive integer $n$, denote $[n] \eqdef \{1,2,\ldots,n\}$. Define $T(d,r)\eqdef(r-1)(d+1)+1$. Tverberg's theorem \cite{tverberg_orig} (see also~\cite{matou_DG}) asserts that for every point sequence $P \eqdef (p_1, \ldots, p_{T(d,r)})$ in $\R^d$ there exists a partition $\mathcal I \eqdef \{I_1,\ldots,I_r\}$ of $[T(d,r)]$ into $r$ parts, such that the $r$ convex hulls $\conv\{p_i \mid i\in I_j\}$, $1\le j\le r$, intersect at a common point. Such a partition $\mathcal I$ is called a \emph{Tverberg partition for $P$}. If the convex hulls intersect at a single point (which happens whenever $P$ is generic), then that point is called the \emph{Tverberg point} of the partition. The special case $r=2$ of Tverberg's theorem is \emph{Radon's lemma}~\cite{radon_orig}; in that case we refer to the Tverberg point as the \emph{Radon point}.

We call a partition of $[T(d,r)]$ into $r$ parts a \emph{Tverberg type}.

Let $\tv_{\mathcal I}(P)$ be the $T(d,r)$-ary predicate stating that the Tverberg type $\mathcal I$ is a Tverberg partition for $P$.

Our main objective in this paper is to classify the Tverberg-type predicates according to the three possibilities mentioned above.

\begin{definition}
Let $\mathcal I \eqdef \{I_1, \ldots, I_r\}$ be a Tverberg type. We call $\mathcal I$ \emph{colorful} if, for each $1\le i\le d+1$, the $r$ consecutive integers $\{(r-1)(i-1)+1,\ldots, (r-1)i+1\}$ belong one to each of the $r$ parts $I_1, \ldots, I_r$.\footnote{Our colorful Tverberg types are unrelated to the \emph{colored Tverberg theorem} (see~\cite{matou_DG}).}
\end{definition}

It is sometimes convenient to encode a Tverberg type as a string $\sigma \in [r]^{T(d,r)}$, indicating to which part each integer belongs. Since the order of the parts within the partition does not matter, there are $r!$ different ways of encoding each Tverberg type. For example, for $d=2$, $r=3$, the Tverberg type $\bigl\{\{1,3,6\},\{2,7\},\{4,5\}\bigr\}$ can be encoded as $1213312$.

In this representation, the colorful Tverberg types are those $\sigma$ for which $\{\sigma(i+1), \ldots, \sigma(i+r)\} = [r]$ for each $i=0,r-1, 2(r-1), \ldots, d(r-1)$. An example of a colorful Tverberg type for $d=3$, $r=5$ is $\underline{1234\overline 5}\overline{241\underline 3}\underline{514\overline 2}\overline{5134}$. The lines above and below the digits indicate the intervals in which all ``colors'' must show up.

It is easily seen that the number of colorful Tverberg types with parameters $d,r$ is $(r-1)!^d$. We will expound on the significance of this number below.

\begin{theorem}\label{thm_stretched-diag}
For every Tverberg type $\mathcal I$, if $\mathcal I$ is colorful, then $\neg \tv_{\mathcal I}$ is avoidable; otherwise, $\tv_{\mathcal I}$ is avoidable.
\end{theorem}

\begin{proof}[Proof sketch]
Take $P$ to be the \emph{stretched diagonal} previously studied in \cite{BMN_epsilonnets,BMN_stabbing,N_thesis}. As we will show in Section~\ref{sec_stair}, the stretched diagonal is homogeneous with respect to all Tverberg-type predicates, and furthermore, the Tverberg types that occur in it are exactly the colorful ones.\footnote{This was noticed independently by Imre B\'ar\'any and Attila P\'or, as well as by Isaac Mabillard and Uli Wagner (private communication).} Since the number of points in the stretched diagonal can be made arbitrarily large, the claim follows.
\end{proof}

\begin{conjecture}\label{conj_colorful}
For every Tverberg type $\mathcal I$, if $\mathcal I$ is colorful, then $\tv_{\mathcal I}$ is unavoidable; otherwise, $\neg \tv_{\mathcal I}$ is unavoidable.
\end{conjecture}

We call the colorful Tverberg type encoded by $12\cdots r\cdots 212\cdots r\cdots$ the \emph{zigzag type}.
For example, the zigzag type for $r=3$ and $d=4$ is $12321232123$.

\begin{theorem}\label{thm_main}
Conjecture~\ref{conj_colorful} holds in the following cases:
\begin{itemize}
\item For $d\le 2$ and all $r$.
\item For $d=3$, for all Tverberg types that have parts of sizes $\{2,3,\allowbreak 4,4,\ldots,4\}$.
\item For $d=r=3$, for all the four colorful Tverberg types that have parts of sizes $\{3,3,3\}$.
\item For the zigzag type for $d=4$ and all $r$.
\end{itemize}
\end{theorem}

The results of Theorem~\ref{thm_main} regarding zigzag types for $d\le 4$ were previously announced in~\cite{BM_decidability}.

\paragraph{Motivation.}
We started studying this topic when we tried to prove that the zigzag type is unavoidable for $r=3$ and all $d$. This was the missing link in our argument that there exist one-sided epsilon-approximants of constant size with respect to convex sets~\cite{BN_onesided}. We managed to prove the zigzag-type claim---and hence, the one-sided epsilon-approximant corollary---only for $d\le 4$. However, we subsequently realized that we could do without the zigzag-type claim, relying instead on what is now Lemma~4 in \cite{BN_onesided}, which is an easier Ramsey-type result and holds for all $d$. 

\subsection{Proof strategy}

Our proof strategy for proving that a given predicate $\Phi$ is unavoidable is as follows: Suppose for a contradiction that for every $n$ there exists an $n$-point sequence $P$ that avoids $\Phi$. Let $\Psi_1, \ldots, \Psi_k$ be other predicates. By Ramsey's theorem, by making $n$ large enough, we can guarantee the existence of arbitrarily large subsequences $P'$ of $P$ that are homogeneous with respect to $\Psi_1, \ldots, \Psi_k$.

Hence, in our search for a contradiction, we can assume without loss of generality that our sequence $P$ not only avoids $\Phi$, but is also homogeneous with respect to a fixed finite family $\overline \Psi \eqdef (\Psi_1, \ldots, \Psi_k)$ of other predicates. Furthermore, if we previously showed that some of the $\Psi_i$'s are unavoidable, then we can assume that $P$ specifically avoids their negations $\neg \Psi_i$.

In our proofs, we will start with an orientation-homogeneous sequence (which corresponds to taking $\Psi_1=\orient$), and we will add additional predicates to $\overline \Psi$ ``on-the-fly''; this is OK as long as we do it only a finite number of times.

We say that predicates $\Phi_1$ and $\Phi_2$ are \emph{equivalent} if there is an unavoidable predicate $\Phi$ such that $\Phi\implies(\Phi_1\iff \Phi_2)$. In our proofs,
$\Phi$ will usually be the orientation predicate. We write $\Phi_1\equiv \Phi_2$ if $\Phi_1$ and $\Phi_2$ are equivalent.

\subsection{Sierksma's conjecture}

Let $P$ be a sequence of $T(d,r)$ points in $\R^d$. Tverberg's theorem asserts the existence of at least one Tverberg partition for $P$. However, usually there is more than one Tverberg partition. If $r=2$ (the case of Radon's lemma), and the points of $P$ are in general position, then the Radon partition is indeed unique. However, it seems that for $r\ge 3$ the Tverberg partition is \emph{never} unique.

Sierksma conjectured in 1979 (\cite{sierksma_cheese}; cited by Reay in~\cite[Problem 14]{reay_problems}) that the number of Tverberg partitions is always at least $(r-1)!^d$. Sierksma pointed out that there exist $T(d,r)$-point sets that have exactly $(r-1)!^d$ Tverberg partitions: Choose $d+1$ affinely independent points $p_1, \ldots, p_{d+1}$, and let $q$ be a point in the interior of the simplex $\conv\{p_1, \ldots, p_{d+1}\}$. Replace each $p_i$ by a tiny cloud $\mathcal P_i$ of $r-1$ points. Then, the Tverberg partitions of this point set are exactly those that have $r-1$ parts containing exactly one point from each cloud, plus an $r$-th part containing only $q$. Hence, the number of Tverberg partitions here equals exactly $(r-1)!^d$. (However, for $d\ge 2$ none of these partitions are colorful, no matter how the points are linearly ordered, since for $d\ge 2$, colorful partitions never contain parts of size $1$.)

White~\cite{white} recently found a more general family of $T(d,r)$-point sets that have exactly $(r-1)!^d$ Tverberg partitions. In fact, he constructs, for every partition $T(d,r) = n_1+\cdots+n_r$ of $T(d,r)$ into $r$ integers satisfying $1\le n_i\le d+1$, a $T(d,r)$-point set $P$ that has exactly $(r-1)!^d$ Tverberg partitions, all of which have parts of sizes $n_1, \ldots, n_r$ and have the origin as their Tverberg point. Furthermore, in his construction, each point $p_i$ is only specified by a vector of signs $v_i \in \{+,0,-\}^d$, which indicates the sign of each coordinate of $p_i$; the magnitudes of the coordinates can be chosen arbitrarily.

Regarding lower bounds, Hell~\cite{hell_number}, extending earlier work of Vu\v ci\'c and Zivaljevi\'c~\cite{VZ_note_sierksma}, showed that the number of Tverberg partitions is always at least $(r-d)!$, and that if $r=p^k$ is a prime power, then the number is at least
\begin{equation*}
\frac{1}{(r-1)!} \left(\frac{r}{k+1}\right)^{\lfloor T(d,r)/2\rfloor}.
\end{equation*}
For large $d$ and $r$, this number is roughly the square root of Sierksma's conjectured bound. 

\paragraph{Our result.} In this paper we construct a broader family of $T(d,r)$-point sets that have exactly $(r-1)!^d$ Tverberg partitions. Our result is a corollary of the proof of Theorem~\ref{thm_stretched-diag}. We show that in stair-convex geometry (previously studied by the authors in \cite{BMN_epsilonnets,BN_centerlines,N_thesis}) \emph{every} generic $T(d,r)$-point set has exactly $(r-1)!^d$ \emph{stair-Tverberg} partitions. As a consequence, in Euclidean geometry, $T(d,r)$ randomly chosen points from the \emph{stretched grid} (\cite{BMN_epsilonnets,BN_centerlines,N_thesis}) will almost surely have exactly $(r-1)!^d$ Tverberg partitions.

\subsection{Concurrent work}

Very recently, and independently, Attila P\'or announced~\cite{por_announcement} that he has found a full proof of Conjecture~\ref{conj_colorful}. P\'or uses a different approach from the one we use in this paper.

\subsection{Organization of the paper}

In Section~\ref{sec_semi} we briefly show that the predicates we consider in this paper are semialgebraic.
In Section~\ref{sec_unavoidable} we prove Theorem~\ref{thm_main}. We first show how solving Tverberg-type predicates reduces to solving \emph{hyperplane-side predicates}. Then we solve many hyperplane-side predicates, first in the plane, then in $d=3$, and then in $d=4$.
In Section~\ref{sec_stair} we prove Theorem~\ref{thm_stretched-diag} and our result regarding Sierksma's conjecture. A key technical lemma that is obvious but whose proof is quite tedious is proven in Appendix~\ref{app_transference}. Finally, Appendix~\ref{app_code} contains Mathematica code for some computer enumerations that we performed.

\section{Our predicates are semialgebraic}\label{sec_semi}

In this section we show that all the predicates defined in the Introduction are semialgebraic.
The orientation predicate is clearly semialgebraic.

\begin{lemma}\label{lem_ptinconv}
Given generic points $q, p_1, \ldots, p_{d+1}\in\R^d$, we have $q \in \conv\{p_1,\allowbreak \ldots,\allowbreak p_{d+1}\}$ if and only if, for each $1\le i\le d+1$, $\orient\{p_1, \ldots, p_{d+1}\}$ equals the orientation obtained by replacing $p_i$ by $q$.
\end{lemma}

\begin{proof}
This follows immediately from the definition of convex combination and Cramer's rule.
\end{proof}

\begin{corollary}\label{cor_ptinconv_semi}
The $(d+2)$-ary predicate ``$q \in \conv\{p_1, \ldots, p_{d+1}\}$'' is semialgebraic.
\end{corollary}

\begin{observation}
All predicates mentioned in the Introduction are semialgebraic.
\end{observation}

\begin{proof}
The four-point planar convex-position predicate can be formulated by stating that none of the four given points lies in the convex hull of the other three. Hence, by Corollary~\ref{cor_ptinconv_semi}, this predicate is semialgebraic.

Now, consider a Tverberg-type predicate $\tv_{\mathcal I}$ for $\mathcal I \eqdef \{I_1, \ldots, I_r\}$. For each $j$, let $x_j$ be an affine combination of the points $p_i$, $i\in I_j$. Hence, there are a total of $T(d,r)$ coefficients in the $r$ affine combinations; these are our unknowns. For each $j$ there is an equation requiring that the coefficients of the $j$-th affine combination add up to $1$. Further, we express the requirement $x_1 = x_2 = \ldots= x_r$ by $(r-1)d$ equations. Hence, the total number of equations is also $T(d,r)$. Therefore, we have a linear system, which has a unique solution if the given points are generic. The unique solution can be expressed using Cramer's rule. Then, the predicate asserts that all the values in this solution are positive, so that the affine combinations are in fact convex combinations.
\end{proof}

\section{Proofs that Tverberg types are unavoidable}\label{sec_unavoidable}

In this section we prove Theorem~\ref{thm_main}. The case $d=1$ is trivial, so let $d\ge 2$.
Recall that a point sequence $P\in(\R^d)^n$ is \emph{orientation-homogeneous} if all size-$(d+1)$ subsequences of $P$ have the same orientation. 
We need the following result due to Gale~\cite{gale}.

\begin{lemma}[Radon-partition lemma]\label{lem_radon}
Let $P = (p_1, \ldots, p_{d+2})$ be orientation-homogeneous. Then the Radon partition for $P$ is the alternating one, i.e. the one encoded by 
$\sigma = 12121\ldots$.
\end{lemma}

\begin{proof}
In general, the Radon partition of a set of $d+2$ points is obtained as follows: We find a nontrivial solution to $\sum \alpha_i = 0$, $\sum \alpha_i p_i = \vec 0$; then, in the latter equation we move all terms with negative $\alpha_i$'s to the right-hand side; finally, we normalize both sides so the sum of coefficients is $1$.

In our case, we add the equation $\alpha_1=1$ to ensure the linear system has a unique, nontrivial solution. Then, Cramer's rule yields that the signs of the $\alpha_i$'s alternate.
\end{proof}

\subsection{Some notation}

To avoid subscripts we shall use integers to denote points. So,
for example, $7$ will denote the $7$th point of sequence $P$. Points after the $9$th are denoted by the letters $A,B,C,\ldots$.
Similarly, we shall write $p_1<p_2$ if point $p_1$ precedes point $p_2$ in $P$.

The convex hull operation will be indicated by the concatenation of the
corresponding integers. So, for example $27$ is the line
segment going from the $2$nd point to the $7$th, and $245$
is the triangle with vertices $2$, $4$ and~$5$.

We shall use two notations for intersection. First, we use
the conventional $179\cap 028$ to denote the intersection of 
triangles $179$ and $028$. Secondly, we denote the same by
placing the two objects to be intersected one above the other, 
like so: $^{179}_{028}$.

Next, we introduce notation for \emph{hyperplane-separation statements}. If $p_1, \ldots, p_d$, $q_1$, $q_2$, $\ldots$ , $r_1$, $r_2$, $\ldots$ are points in $\R^d$, then we write
\begin{equation*}
p_1\cdots p_d(q_1q_2\cdots:r_1r_2\cdots)
\end{equation*}
to mean that the hyperplane spanned by $p_1, \ldots, p_d$ separates $q_1, q_2, \ldots$, from $r_1, r_2, \ldots$. In other words, $\orient(p_1, \ldots, p_d, q_i)$ has the same value for all $i$, which is the opposite of $\orient(p_1, \ldots, p_d, r_i)$ for all $i$. For instance, $148(2:7_{258}^{\,37})$ means that $7$ and $_{258}^{\,37}$
are on one side of the hyperplane $148$, whereas $2$ is on the other. (Of course, in order for $1,4,8$
to span a hyperplane and for $_{258}^{\,37}$ to be a Radon point, we must be in $\R^3$.)

Given a hyperplane-separation statement $s \eqdef H(p:q)$ involving $k$ distinct points, we denote by $\Pi[s]$ the corresponding $k$-ary geometric predicate. For example, in $\R^2$, consider the two statements $s_1 \eqdef 25(1:{}^{14}_{36})$ and $s_2 \eqdef 37(1:{}^{15}_{49})$. Then $\Pi[s_1] \equiv \Pi[s_2]$; both denote the $6$-ary semialgebraic predicate that asserts, given a length-$6$ orientation-homogeneous sequence $(p_1, \ldots, p_6)$, that the line through $p_2$ and $p_5$ separates $p_1$ from $p_1p_4 \cap p_3p_6$; the latter intersection point exists if $p_1, \ldots p_6$ are orientation-homogeneous by the Radon-partition lemma (Lemma~\ref{lem_radon}). (If the input points are not orientation-homogeneous, then we do not care what the predicate asserts.)

\begin{lemma}
Let $P$ be orientation-homogeneous. Let $p_1 < p_2 < \cdots < p_d$ be points of $P$, and let $q < q'$ be two other points of $P$. Then, we have $p_1 p_2 \cdots p_d(q:q')$ if and only if the number of $p_i$'s between $q$ and $q'$ is odd; otherwise, we have $p_1 p_2 \cdots p_d (qq':{})$.
\end{lemma}

\begin{proof}
Recall every interchange between two columns of a matrix multiplies the value of its determinant by $-1$. Then the claim follows by counting the number of column interchanges.
\end{proof}

Thus, for example, in $\R^4$ we have $1368(27:459)$.

\begin{observation}[Same-side rule]\label{obs_sameside}
If a simplex $\sigma$
lies entirely on one side of a hyperplane $H$, then any intersection $\sigma\cap \tau$ lies
on the same side of $H$, for any set $\tau$.
\end{observation}

For example, in $\R^2$ we have $14(5:{}^{13}_{25})$, since the segment $13$ is entirely on one side of the line through $14$, specifically, on the side opposite to $5$.

\subsection{The six-point lemma}

\begin{figure}
\centerline{\includegraphics{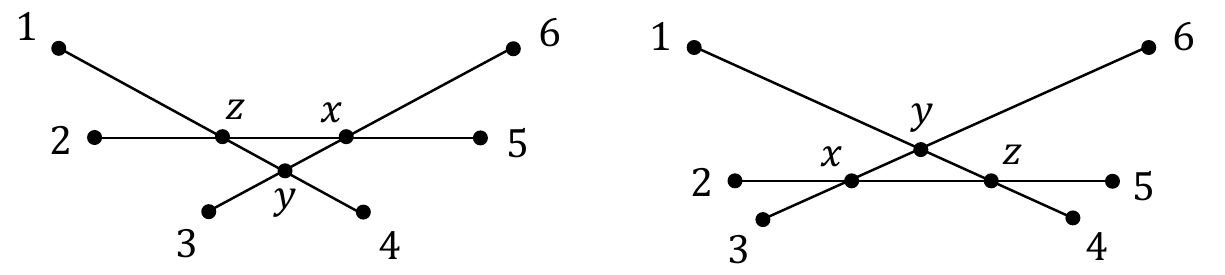}}
\caption{\label{fig_sixpt}In both figures points $1,\ldots,6$ are orientation-homogeneous. However, in the left figure, the statements of Observation~\ref{obs_sixequiv} hold, whereas in the right figure they do not hold.}
\end{figure}

\begin{observation}\label{obs_sixequiv}
Let points $1,\ldots, 6$ be orientation-homogeneous in the plane. Let $x\eqdef 25\cap 36$, $y\eqdef 14\cap 36$, $z\eqdef 14\cap 25$. Then the statements $14(3:x)$, $25(1:y)$, $36(4:z)$ are all equivalent. (See Figure~\ref{fig_sixpt}.)
\end{observation}

\begin{proof}
Suppose $14(3:x)$. Then, along the segment $36$, the points $3,y,x,6$ lie in this order. Hence, $25(6:y)$, which is equivalent to $25(1:y)$.

Now suppose $25(1:y)$. Therefore, along the segment $14$ the order is $1,z,y,4$. Hence, $36(4:z)$.

Finally, suppose $36(4:z)$, which is equivalent to $36(5:z)$. Then, along the segment $25$ the order is $2,z,x,5$. Hence, $14(2:x)$, which is equivalent to $14(3:x)$.
\end{proof}

\begin{lemma}[Six-point lemma]\label{lem_sixpt}
The planar predicate 
\begin{equation*}
\sixpt\eqdef \Pi\bigl[14\bigl(3:{}^{25}_{36}\bigr)] \quad \Bigl({} \equiv \Pi\bigl[25\bigl(1:{}^{14}_{36}\bigr)\bigr] \equiv \Pi\bigl[36\bigl(4:{}^{14}_{25}\bigr)\bigr] \Bigr)
\end{equation*}
of Observation~\ref{obs_sixequiv} is unavoidable.
\end{lemma}

\begin{proof}
Let $1,\ldots, 7$ be orientation-homogeneous. We will show that we cannot have both $\neg\sixpt(1,2,3,\allowbreak 4,5,6)$ and $\neg\sixpt(2,3,4,\allowbreak 5,6,7)$.

\begin{figure}
\centerline{\includegraphics{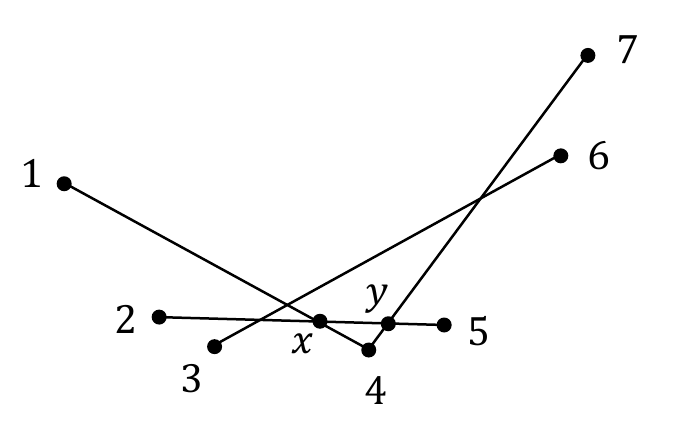}}
\caption{\label{fig_R2}Proof of the six-point lemma (Lemma~\ref{lem_sixpt}).}
\end{figure}

Let $x\eqdef 14\cap 25$, $y\eqdef 25\cap 47$. Suppose that $\neg\sixpt(1, \ldots, 6)$, or in other words, that $36({}:x5)$. However, by the same-side rule (Observation~\ref{obs_sameside}), we have $14({}:y5)$, so the order along the segment $25$ is $2,x,y,5$. In particular, $y\in x5$. Therefore, by the same-side rule we have $36({}:y5)$, which is equivalent to $\sixpt(2,\ldots, 7)$. See Figure~\ref{fig_R2}.
\end{proof}

\subsection{The case \texorpdfstring{$\boldsymbol{d=2}$}{d=2}}

\begin{lemma}
All the planar colorful Tverberg types are unavoidable.
\end{lemma}

\begin{proof}
We first show that the general case $r\ge 3$ reduces to the case $r=3$.

Let $r\ge 3$. Then, every planar colorful Tverberg type $\mathcal I$ into $r$ parts is encoded by a string of the form $\sigma\eqdef \pi_1 1 \pi_2 2 \pi_3$, where $\pi_1$ is a permutation of $[r]\setminus\{1\}$, $\pi_2$ is a permutation of $[r]\setminus\{1,2\}$, and $\pi_3$ is a permutation of $[r]\setminus\{2\}$. Hence, parts $1$ and $2$ have size $2$, while every part $i\ge 3$ has size $3$. The predicate $\tv_{\mathcal I}$ asserts that the segments corresponding to parts $1$ and $2$ intersect at a point $x$, which is contained in all the triangles corresponding to parts $i\ge 3$. So, it suffices to show that $x$ is contained in the convex hull of part $3$; by symmetry, it would then follow
that $x$ is also contained in the convex hull of part $i$ for each $i\ge 3$.

Let $3\le i\le r$. For simplicity assume $i=3$. Then, the restriction of $\sigma$ to $\{1,2,3\}$ is of the form $\sigma|_{\{1,2,3\}} = \{2,3\}132\{1,3\}$ (where $\{a,b\}$ means either $ab$ or $ba$). These are exactly the encodings of the four colorful types with $r=3$.

\begin{figure}
\centerline{\includegraphics{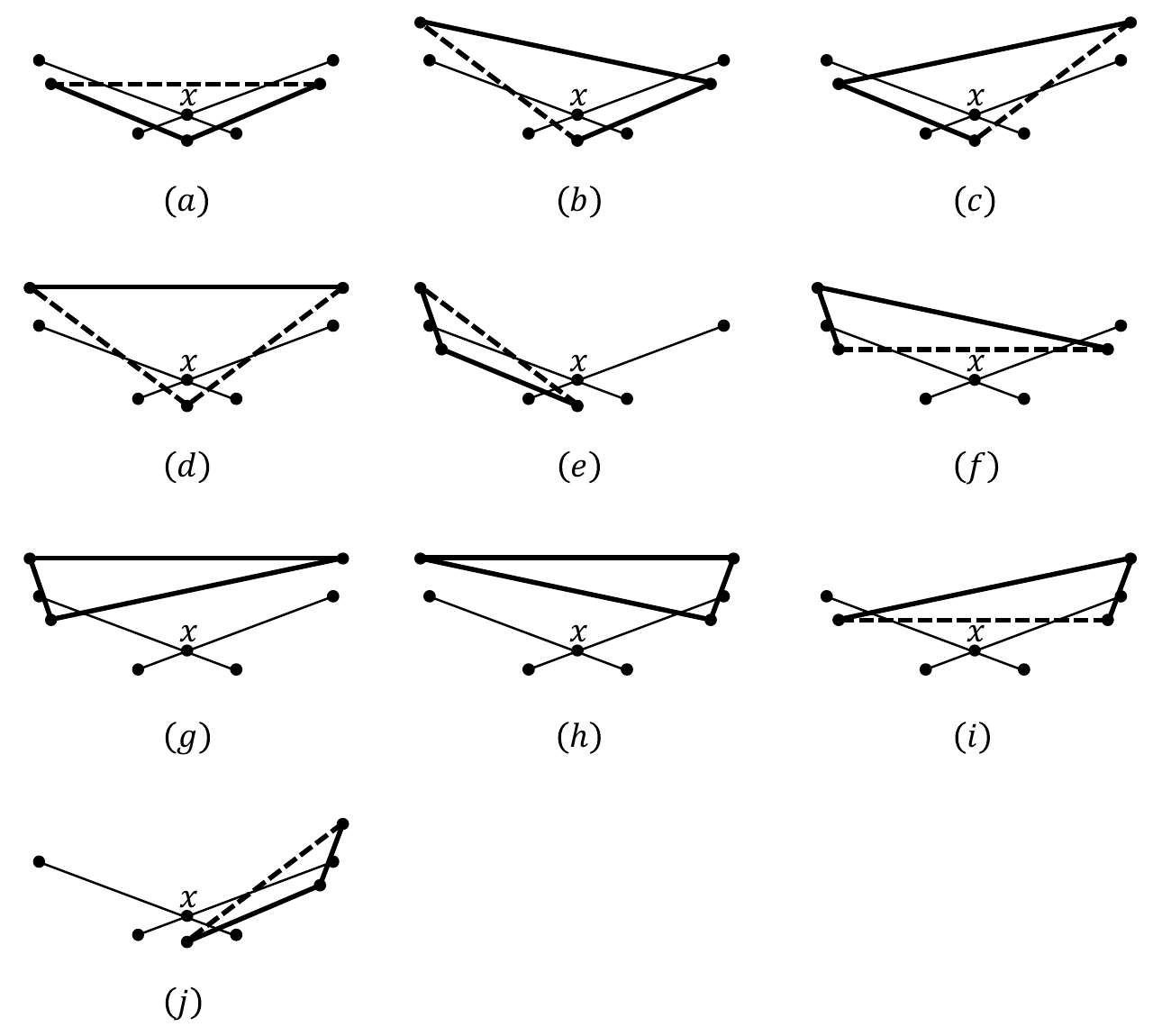}}
\caption{\label{fig_TvR2}Case analysis for $d=2$. In each case, the position with respect to $x$ of the solid edges is given by the same-side rule, whereas the position with respect to $x$ of the dotted edges is given by the six-point lemma.}
\end{figure}

By Lemma~\ref{lem_ptinconv}, each corresponding predicate is a conjunction of three line-separation predicates (assuming the given points are orientation-homogeneous). For example, for the case $\sigma|_{\{1,2,3\}} = 3213231$, the corresponding line-separation predicates are
\begin{equation*}
\Pi[16({}:x4)], \qquad \Pi[14({}:x6)], \qquad \Pi[46({}:x1)],
\end{equation*}
where $x\eqdef 25\cap 37$. The first and third predicates hold in any orientation-homogeneous sequence by the same-side rule. And the second predicate is equivalent to $\sixpt$, which we showed in Lemma~\ref{lem_sixpt} to be unavoidable. See Figure~\ref{fig_TvR2}($b$).

Similarly, in the other three possible values for $\sigma_{\{1,2,3\}}$, there are one or two predicates that are unavoidable by the same-side rule, while the remaining ones are equivalent to $\sixpt$. See Figure~\ref{fig_TvR2}($a$,$c$,$d$).
\end{proof}

We now proceed to show that, for all planar non-colorful Tverberg types, their negation is unavoidable. For this, we first prove a lemma that will be very useful in higher dimensions as well:

\begin{lemma}[No-consecutive-points lemma]\label{lem_no-consecutive}
If $\mathcal I \eqdef \{I_1, \ldots, I_r\}$ is a Tverberg type in which some part contains two consecutive integers, then $\neg \tv_{\mathcal I}$ is unavoidable.
\end{lemma}

\begin{proof}
Suppose without loss of generality that $\{a, a+1\} \subseteq I_1$. Let $k \eqdef |I_1|$.

Suppose first that $2\le k \le d$. Let $n \eqdef T(d,r)+1$, and suppose for a contradiction that the orientation-homogeneous point sequence $P \eqdef (p_1,\ldots, p_n)$ avoids $\neg \tv_{\mathcal I}$. Let
\begin{equation*}
P' = \{p_b : b\in I_1 \wedge b\le a\} \cup \{p_{b+1} : b\in I_1 \wedge b\ge a\}
\end{equation*}
(so, for example, if $I_1 = \{1,3,5,6,8,11\}$ and $a=5$, then $P' = \{p_1,p_3,p_5,p_6,p_7,p_9,p_{12}\}$). Let $P'_i \eqdef P' \setminus \{p_{a+i}\}$ for $i=0,1,2$. When we evaluate $\tv_{\mathcal I}$ at the points $P\setminus \{p_{a+i}\}$ for $i=0,1,2$, part $I_1$ is assigned the points of $P'_i$, whereas the remaining parts are assigned points independently of $i$.

Let $f$ be the intersection of the affine hulls of the parts $I_j$, $j\ge 2$; hence, $f$ is a $(d-k+1)$-dimensional flat. The affine hull of $P'$ intersects $f$ at a line $\ell$. However, this line $\ell$ can intersect the interior of at most two of the convex hulls of $P'_i$, $i=0,1,2$, since they are three distinct faces of the simplex spanned by $P'$. Contradiction.

If $k=d+1$ we use a different argument: Recall that by Lemma~\ref{lem_ptinconv}, $\tv_{\mathcal I}$ reduces to a Boolean combination of $d+1$ predicates of the form $\Pi(b) \eqdef {}$``the intersection point of parts $I_2, \ldots, I_r$ lies on the positive side of the hyperplane $I_1\setminus\{b\}$,'' for each $b\in I_1$. Hence, $\Pi(a)$ and $\Pi(a+1)$ are equivalent predicates. We can assume that the given point sequence $P$ is homogeneous with respect to it. But, in order for $\tv_{\mathcal I}$ to hold, $\Pi(a)$ and $\Pi(a+1)$ must have opposite values. 
\end{proof}

\begin{lemma}
For every non-colorful Tverberg type $\mathcal I$ in the plane, $\neg \tv_{\mathcal I}$ is unavoidable.
\end{lemma}

\begin{proof}
If one of the parts in $\mathcal I$ has size $1$, or two of the parts have size $2$ but they do not alternate, then, by the Radon-partition lemma (Lemma~\ref{lem_radon}), $\tv_{\mathcal I}$ does not hold in an orientation-homogeneous sequence.

Hence, suppose parts $I_1$ and $I_2$ alternate, so the encoding of $\mathcal I$ has the form $\ldots 1\ldots 2\ldots 1\ldots 2\ldots$, partitioning the interval $[T(d,r)]$ into five gaps. Consider a part $I_j$, $j\ge 3$. Assume $j=3$ for simplicity. By the no-consecutive-points lemma (Lemma~\ref{lem_no-consecutive}), we can rule out all cases in which two $3$'s belong to the same gap. Hence, there are only a few cases left, which can be easily ruled out. See Figure~\ref{fig_TvR2}($e$--$j$).
\end{proof}

\subsection{Central projection}

We now present a technique for lifting results to higher dimensions.
Let $H$ be a hyperplane in $\R^d$. To define an orientation predicate within $H$, fix a rigid motion $\rho$ in $R^d$ that takes $H$ to the hyperplane $H_0$ given by $x_d=0$. Then define $\orient_H$ by $\orient_H(p_1, \ldots, p_d) \eqdef \orient(\rho(p_1), \ldots, \rho(p_d))$ (where in the last predicate we ignore the last coordinates of the points, which equal $0$).

Fix a point $q\in\R^d\setminus H$. Each choice of $\rho$ produces one of two possible predicates $\orient_H$, which are negations of one another, depending on whether $\rho$ sends $q$ above or below $H_0$. Let us fix a $\rho$ that sends $q$ above $H_0$. Call the resulting predicate $\orient_{H,q}$. It is not hard to see that
\begin{equation}\label{eq_orientH}
\orient_{H,q}(p_1, \ldots, p_d) = \orient(p_1, \ldots, p_d, q).
\end{equation}

Given a point set $X\subset\R^d$ and a point $p \notin \conv X$, fix a hyperplane $H$ that separates $p$ from $X$. Then we define the \emph{central projection} of $X$ from $p$ into $H$ by taking each point $q\in X$ to the intersection point $q'\eqdef pq \cap H$; and we define the orientation within $H$ by $\orient_{H,p}$.

\begin{observation}
Let $p_1, ..., p_n \in \R^d$ be orientation-homogeneous, and let $H$ be a hyperplane that separates $p_1$ (resp.~$p_n$) from the rest of the points. Then, centrally-projecting the rest of the points from $p_1$ (resp.~$p_n$) into $H$ produces an orientation-homogeneous sequence in $H$.
\end{observation}

\begin{proof}
By (\ref{eq_orientH}).
\end{proof}

Note that projection from an intermediate point $p_i$ does not produce an orientation-homogeneous sequence, since then the orientation of a $d$-tuple of projected points will depend on the parity of the number of points that appear after $p_i$.

\begin{observation}
Let $p_1,...,p_n \in \R^d$, and let $p$ be another point not in $\conv{\{p_1,...,p_n\}}$. Let their central projection from $p$ into a hyperplane $H$ be $p'_1,...,p'_n$ respectively. Then the central projection of $\conv{\{p_1,...,p_n\}}$ equals $\conv{\{p'_1,...,p'_n\}}$.
\end{observation}

\begin{corollary}\label{cor_proj}
Let $P = (p_0, p_1,\ldots,p_{d+1}) \in (\R^d)^{d+2}$, with $p_0\notin\conv\{p_1,\allowbreak \ldots,\allowbreak p_{d+1}\}$. Let $P' = (p'_1,...,p'_{d+1})$ be the central projection of $\{p_1, \ldots, p_{d+1}\}$ from $p_0$ into a hyperplane $H$. Then:
\begin{itemize}
\item The Radon point of $P'$ is the central projection of the Radon point of $P$.

\item Let $[d+1]=I_1\cup I_2$ be the Radon partition of $P'$. Then the Radon partition of $P$ is either $I_1\cup\{0\},I_2$ or $I_1,I_2\cup\{0\}$.
\end{itemize}
\end{corollary}

Hence, we can lift up lower-dimensional results by adding a new point at the beginning 
or at the end. For example:

\begin{lemma}\label{lem:3d}
In $\R^3$, the predicates 
\begin{equation}\label{eq:3d}
\Pi[025(1:{}^{\,14}_{036})] \text{\quad and \quad} \Pi[257(1:{}^{147}_{\,36})]
\end{equation}
are unavoidable.
\end{lemma}

\begin{proof}
For the first predicate, let points $0, \ldots, 7$ be orientation-homogeneous in $\R^3$, let $\pi$ be the plane through $0,2,5$, and let $r$ be the ray emanating from $0$ and passing through the point $14\cap 036$. This ray $r$ is entirely on one side of $\pi$, and the question is on which side. Project points $1,\ldots,7$ centrally from $0$ into a plane $H$, obtaining points $1',\ldots,7'$. The question reduces to whether $2'5'(1':{}^{1'4'}_{3'6'})$ holds in $H$. Since the the planar predicate $\sixpt = \Pi[25(1:{}^{14}_{36})]$ is unavoidable, we conclude that $\Pi[025(1:{}^{\,14}_{036})]$ in $\R^3$ is also unavoidable.

Similarly, when we project centrally from $7$, we obtain that $\Pi[257(1:{}^{147}_{\,36})]$ is unavoidable. (Alternatively, the right predicate is the mirror image of the left one.)
\end{proof}

\begin{corollary}\label{cor:3dcombined}
The predicate
\begin{equation}\label{eq:3dcombined}
\Pi[147(3:{}^{258}_{\,36})]
\end{equation}
is unavoidable.
\end{corollary}

\begin{proof}
Let $P = (1,\ldots,9)$ be a point sequence in $\R^3$ that is orientation-homogeneous, as well as homogeneous with respect to the two predicates in (\ref{eq:3d}). Let us reformulate the left predicate in (\ref{eq:3d}): This predicate asserts that the Radon points $x\eqdef 025\cap 14$ and $y\eqdef 036\cap 14$ lie along the segment $14$ in the order $1,x,y,4$. Hence, the predicate is equivalent to $\Pi[036(1{}^{025}_{\,14}:{})]$. In particular, in $P$ we have $159(3{}^{148}_{\,37}:{})$.

Similarly, the right predicate in (\ref{eq:3d}) is equivalent to $\Pi[147(3:{}^{257}_{\,36})]$. Hence, in $P$ we have $159(3:{}^{269}_{\,37})$.

Hence, if we let $u\eqdef 148\cap 37$, $v\eqdef 159\cap 37$, $w\eqdef 269\cap 37$, the order along $37$ is $3,u,v,w,7$. Therefore, $148(3:w)$. This is an instance of predicate (\ref{eq:3dcombined}).
\end{proof}

Note that there are many plane-side predicates in $\R^3$ involving a Radon point, which are not covered by the above result, nor are they trivially solved by the same-side rule; for example, $\Pi[345({}:{}^{168}_{\,27})]$ and $\Pi[368({}:{}^{147}_{\,25})]$. Below in Section~\ref{subsec_d3} we will solve the latter one (and its mirror image).

\subsection{A movement interpretation}\label{subsec_movement}

We now give a useful way of reinterpreting Lemma~\ref{lem:3d} and Corollary~\ref{cor:3dcombined}.

Let $P\in(\R^3)^n$ be orientation-homogeneous, and let $a_1,\allowbreak a_2,a_3,\allowbreak a'_1,a'_2,\allowbreak a'_3,b_1,\allowbreak b_2 \in P$, such that $b_1, b_2$ interlace both $a_1,a_2,a_3$ and $a_1',a_2',a_3'$, meaning, $a_1<b_1<a_2<b_2<a_3$ and $a_1'<b_1<a_2'<b_2<a_3'$. Define the Radon points
\begin{equation*}
r\eqdef {}^{\ b_1b_2}_{a_1a_2a_3}, \qquad r'\eqdef {}^{\ b_1b_2}_{a_1'a_2'a_3'}.
\end{equation*}
Hence, both $r$ and $r'$ lie along the segment $b_1b_2$. The question is in which order.

If $a'_1 \le a_1$ and $a'_2\ge a_2$ and $a'_3 \le a_3$, then, by the same-side rule (Observation~\ref{obs_sameside}), it is immediate that the order is $b_1,r,r',b_2$.

However, what happens if there are ``conflicting'' movements, e.g.~if $a'_1>a_1$ and $a'_2>a_2$? In fact, in such cases both outcomes are possible. However, one of the outcomes is unavoidable. Specifically:

\begin{lemma}[Movement lemma]\label{lem_movement}
Let $a_1,a_2,\dotsc,b_1,b_2,r,r'$ be defined as above, and suppose $a'_2 > a_2$. Then the predicate ``The order along $b_1b_2$ is $b_1,r,r',b_2$'' is unavoidable.
In other words, the movement of the point $a_2$ is the decisive one.
\end{lemma}

\subsection{The case \texorpdfstring{$\boldsymbol{d=3}$}{d=3}}\label{subsec_d3}

Unlike in the planar case, for $d=3$ we have incomplete results.

Every colorful Tverberg type with $d=3$ is either of the form
\begin{equation}\label{eq_colorful3_32}
\pi_1 1 \pi_2 2 \pi_3 1 \pi_4,
\end{equation}
where $\pi_1,\pi_4$ are permutations of $[r]\setminus\{1\}$ and $\pi_2,\pi_3$ are permutations of $[r]\setminus\{1,2\}$; or of the form
\begin{equation}\label{eq_colorful3_333}
\pi_1 1 \pi_2 2 \pi_3 3 \pi_4,
\end{equation}
where $\pi_1$ is a permutation of $[r]\setminus\{1\}$, $\pi_2$ is a permutation of $[r]\setminus\{1,2\}$, $\pi_3$ is a permutation of $[r]\setminus\{2,3\}$, and $\pi_4$ is a permutation of $[r]\setminus\{3\}$.

In the case (\ref{eq_colorful3_32}), part $1$ has size $2$, part $2$ has size $3$, and all the other parts have size $4$. In the case (\ref{eq_colorful3_333}), parts $1$, $2$, and $3$ have size $3$, and all the other parts have size $4$.

\paragraph{Parts of sizes 2,3,4.}
Let $P \eqdef (1,\ldots, 9,A,B)$. Assume $P$ is orientation-homogeneous, as well as homogeneous with respect to the predicates covered by the movement lemma (Lemma~\ref{lem_movement}). We will consider the order of various points on the line segment $48$. Let $I_1 \eqdef \{4,8\}$, $I_2 \eqdef \{2,6,A\}$, and let their Radon point be $x\eqdef 48\cap 26A$. In order for $\{I_1, I_2,I_3\}$ to be a colorful Tverberg type, $I_3$ must contain points $5$ and $7$, as well as one of the points $1,3$ and one of the points $9,B$. Hence, there are four colorful Tverberg types. Each one decomposes into four plane-separation predicates by Lemma~\ref{lem_ptinconv}; out of these four plane-separation predicates, two hold by the same-side rule, while the remaining two hold by the movement lemma. For example, for the case $I_3=\{3,5,7,B\}$, the corresponding plane-separation predicates are
\begin{equation*}
\Pi[357({}:xB)], \qquad \Pi[35B({}:x7)], \qquad \Pi[37B({}:x5)], \qquad \Pi[57B({}:x3)].
\end{equation*}
The first and last predicate hold by the same-side rule, while the middle two predicates hold by the movement lemma.

Next, we show that for each non-colorful Tverberg type with parts of sizes $\{2,3,4\}$, its negation is unavoidable. The Radon-partition lemma takes care of all cases in which the parts of sizes $2$ and $3$ do not alternate.\footnote{It also takes care of all cases in which a part has size $1$.} Further, by the no-consecutive-points lemma (Lemma~\ref{lem_no-consecutive}), it is enough to consider those cases where $I_1 \eqdef \{4,8\}$, $I_2 \eqdef \{2,6,A\}$, and $I_3\subset\{1,3,5,7,9,B\}$. Hence, the number of cases is $\binom{6}{4} - 4 = 11$. All cases can be solved using either the same-side rule or the movement lemma. As before, let $x \eqdef 48\cap 26A$.

If $I_3=\{1,3,5,7\}$, then by the same-side rule, $357(x:1)$. If $I_3=\{1,3,5,9\}$, then by the movement lemma, $159(x:3)$. If $I_3 = \{1,3,5,B\}$, then, by the movement lemma, $15B(x:3)$. If $I_3=\{1,3,7,9\}$, then, by the movement lemma, $379(x:1)$. If $I_3 = \{1,3,7,B\}$, then, by the movement lemma, $37B(x:1)$. If $I_3 = \{1,3,9,B\}$, then, by the same-side rule, $139(x:B)$. The remaining five cases are the mirror images of the first five.

\paragraph{Parts of sizes 3,3,3.}
We start by handling the four colorful Tverberg types of this form.

\begin{lemma}\label{lem_triplept}
The following four colorful Tverberg types in $\R^3$ are unavoidable:
\begin{align*}
&\{1,4,7\},\{2,5,8\},\{3,6,9\}, &\{2,4,7\},\{1,5,8\},\{3,6,9\},\\
&\{1,4,7\},\{2,5,9\},\{3,6,8\}, &\{2,4,7\},\{1,5,9\},\{3,6,8\}.
\end{align*}
In addition, the following two plane-side predicates are unavoidable:
\begin{equation}\label{eq_moreplanesidepreds}
\Pi[368(4:{}^{147}_{\,25})], \qquad \Pi[136(5:{}^{258}_{\,47})].
\end{equation}
\end{lemma}

\begin{proof}
Consider the first type. Denote $T_1\eqdef 147$, $T_2 \eqdef 258$, $T_3\eqdef 369$. Let $w\eqdef 147\cap 25$, $x\eqdef 47\cap 258$, $y\eqdef 147\cap 36$, $z\eqdef 47\cap 369$. Hence, $T_1\cap T_2 = wx$, and $T_1 \cap T_3 = yz$. Let us see where the points $w$, $x$, $y$, $z$ lie within $T_1$.

\begin{figure}
\centerline{\includegraphics{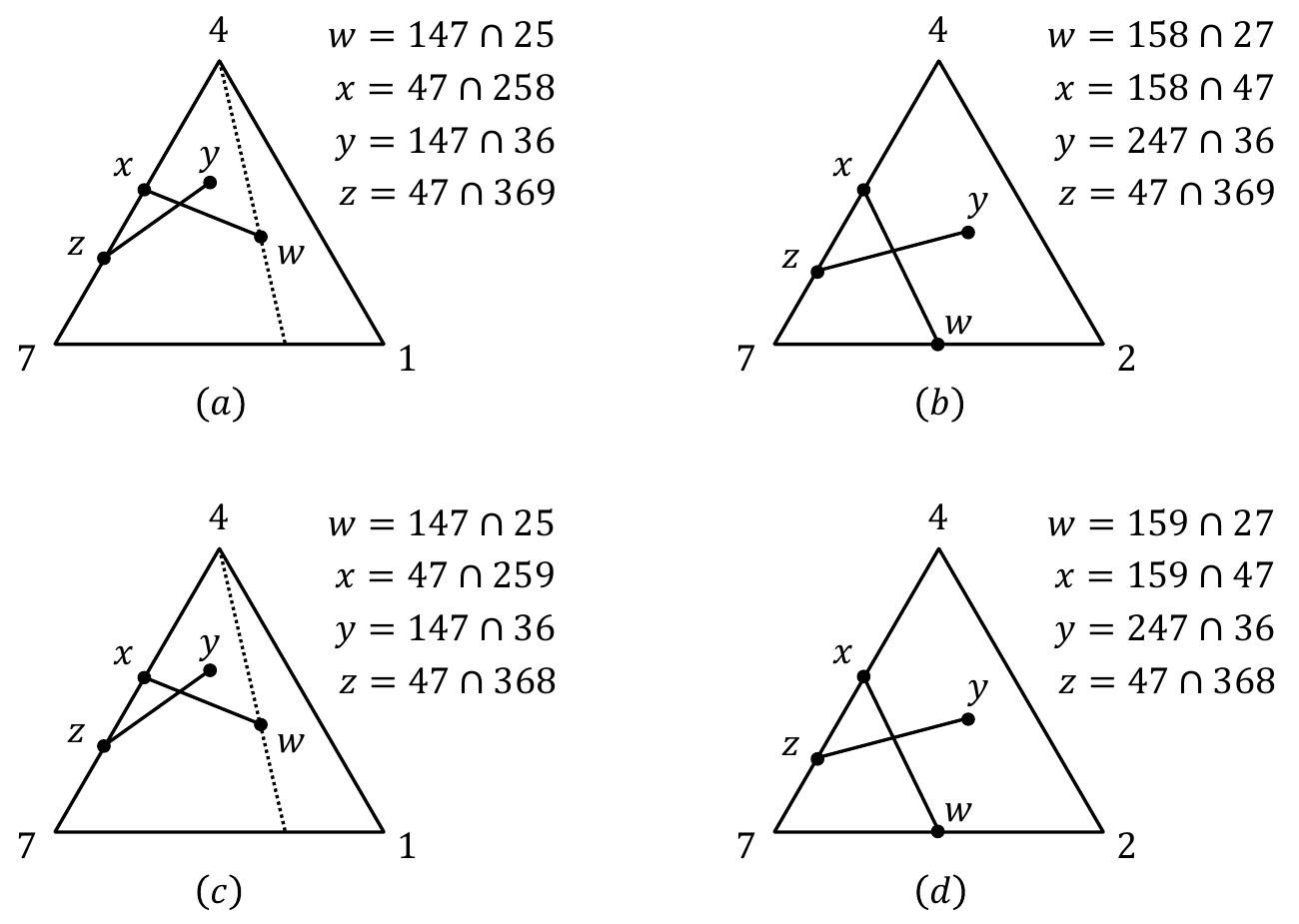}}
\caption{\label{fig_333}($a$) Case $\{1,4,7\},\{2,5,8\},\{3,6,9\}$. ($b$) Case  $\{2,4,7\},\allowbreak \{1,5,8\},\allowbreak \{3,6,9\}$. ($c$) Case  $\{1,4,7\},\allowbreak \{2,5,9\},\allowbreak \{3,6,8\}$. ($d$) Case $\{2,4,7\},\allowbreak \{1,5,9\},\allowbreak \{3,6,8\}$.}
\end{figure}

The points $x$ and $z$ lie on the segment $47$. Furthermore, by the movement lemma (Lemma~\ref{lem_movement}), their order is $4,x,z,7$. The point $w$ lies somewhere in the interior of $T_1$. Where does $y$ lie with respect to $w$? By the same-side rule, we have $245(y7:{})$. Since the plane $245$ intersects $T_1$ along the line $4w$, it follows that within $T_1$ we have $4w(7y:1)$. Furthermore, by the movement lemma, we have $258(4y:17)$, so within $T_1$ we have $wx(4y:17)$. Hence, the situation is as in Figure~\ref{fig_333}($a$), so the segments $wx$ and $yz$ indeed intersect.

We further see from Figure~\ref{fig_333}($a$) that the plane through $T_3$ separates $w$ from $4$. Hence, $369(4:{}^{147}_{\,25})$. Renaming the $9$ to an $8$ we obtain the first predicate in (\ref{eq_moreplanesidepreds}). The second predicate in (\ref{eq_moreplanesidepreds}) is its mirror image.
The other types are handled in a similar way; see Figure~\ref{fig_333}($b$--$d$).
\end{proof}

Out of the $280$ partitions of $\{1,\ldots,9\}$ into three parts of size $3$, a computer enumeration shows that there are $17$ in which all three pairs of triangles intersect (see code in Appendix~\ref{app_code}). Four of them are the colorful Tverberg types listed above in Lemma~\ref{lem_triplept}. Of the $13$ remaining ones, six are solved by the no-consecutive-points lemma (Lemma~\ref{lem_no-consecutive}). The seven remaining Tverberg types are
\begin{align*}
&\{1, 4, 7\}, \{2, 6, 9\}, \{3, 5, 8\}, \qquad &\{1, 4, 8\}, \{2, 6, 9\}, \{3, 5, 7\},\\
&\{1, 4, 9\}, \{2, 5, 7\}, \{3, 6, 8\}, \qquad &\{1, 4, 9\}, \{2, 6, 8\}, \{3, 5, 7\},
\end{align*}
and their mirror images. The avid reader is invited to try to solve them.

\paragraph{Parts of sizes 3,3,3,4.}
A computer enumeration shows that there are $144$ colorful Tverberg types with part sizes $3,3,3,4$. By Lemma~\ref{lem_ptinconv}, each one is a conjunction of four triple-point plane-side predicates of the form ``The intersection $abc\cap def\cap ghi$ is on such a side of the plane $xyz$''. The total number of distinct such plane-side predicates is $240$ (according to our computer program).

Many of these predicates can be proven unavoidable by a straightforward use of the same-side rule. Consider, for example, the predicate $47A(5:x)$ for $x \eqdef 159\cap 26C\cap 38B$. The intersection of the triangles $159$ and $38B$ equals the segment $yz$, where $y\eqdef 159\cap 38$ and $z\eqdef 59\cap 38B$. By the same-side rule we have $47A(5:y)$, and by the movement lemma, $47A(5:z)$. Therefore, since $x$ lies along the segment $yz$, we have, again by the same-side rule, $47A(5:x)$.

There are predicates that do not yield to such simple analysis; for example, $48B(1:x)$ for $x\eqdef 16A\cap 259\cap 37C$. One predicate is highly symmetric, so it can be solved with a trick similar to the one for $\sixpt$:

\begin{lemma}\label{lem_highly_symm}
The predicate $\Pi[48C(5:x)]$ for $x\eqdef 159\cap 26A\cap 37B$ is unavoidable.
\end{lemma}

\begin{proof}
Let $P = (1,\ldots, 9, A, \ldots, D)$. Define the triangles
\begin{align*}
T_1 &\eqdef 159,& T_2 &\eqdef 26A,& T_3 &\eqdef 37B,\\
T_4 &\eqdef 48C,& T_5 &\eqdef 59D.
\end{align*}
Hence, $x=T_1\cap T_2\cap T_3$. Define the triple-intersection points
\begin{equation*}
w \eqdef T_2 \cap T_3 \cap T_4,\qquad y \eqdef T_2 \cap T_3 \cap T_5.
\end{equation*}
Define the intersection points
\begin{equation*}
a \eqdef T_2 \cap 37,\qquad b \eqdef T_3 \cap 6A.
\end{equation*}
So the intersection of the triangles $T_2$ and $T_3$ equals the segment $ab$.

The points $x,w,y$ all lie within the segment $ab$. But in which order do they lie?

\begin{claim}
Along $ab$ we have the order $a,x,y,b$.
\end{claim}

\begin{proof}
By the movement lemma we have $T_1(D:a)$. Furthermore, by the same-side rule, we have $T_1({}:yD)$. Since $T_1$ passes through $x$, the claim follows.
\end{proof}

Now, suppose for a contradiction that $P$ avoids $\Pi[48C(5:x)]$. This means that $T_4(x:A)$, as well as $T_5(w:B)$---here is where the symmetry of the predicate comes into play.

Now, $T_4(x:A)$ implies that, along the segment $ab$, the order is $a,w,x,b$. (Proof: By the movement lemma, we have $T_4(b:A)$, so $T_4(bx:{})$. But $T_4$ passes through $w$.)

From the orders $a,w,x,b$ and $a,x,y,b$ follows the order $a,w,y,b$.

However, $T_5(w:B)$ implies the opposite order $a,y,w,b$! (Proof: By the movement lemma, we have $T_5(b:B)$, so $T_5(wb:{})$. But $T_5$ passes through $y$.)

This contradiction concludes the proof.
\end{proof}

\subsection{The case \texorpdfstring{$\boldsymbol{d=4}$}{d=4}}

In this section we prove that the zigzag type for $d=4$, $r=3$ is unavoidable. Then it follows immediately that the same is true for all $r\ge 3$. Recall that this Tverberg type is encoded by $12321232123$, and it is given by $I_1 \eqdef \{1,5,9\}$, $I_2 \eqdef \{2,4,6,8,A\}$, $I_3 \eqdef \{3,7,B\}$. Let $x \eqdef 159\cap 37B$ be the Radon point of parts $I_1$ and $I_3$.

By Lemma~\ref{lem_ptinconv}, in order to show that the simplex spanned by $I_2$ contains $x$, we have to show that the following five hyperplane-side statements are unavoidable:
\[
  \def\qspace{\hskip 0.75em plus 0.25em minus 0.25em} 
  2468(Ax:{}),\qspace 246A(8x:{}),\qspace 248A(6x:{}),\qspace 268A(4x:{}),\qspace 468A(2x:{}).
\]
By symmetry,
we only have to deal with the first three statements. The first statement follows immediately from the same-side rule.

\begin{lemma}\label{lem_4Da}
The predicate $\Pi[(246A(8{}^{159}_{37B}:{})]$  is unavoidable.
\end{lemma}

\begin{proof}
A predicate remains the same if we relabel points preserving their relative order. So, we rewrite the predicate as $\Pi[(246A(7{}^{159}_{37B}:{})]$ and then as $\Pi[(2469(7x:{})]$ where $x\eqdef 158\cap 37A$.

\begin{figure}
\centerline{\includegraphics{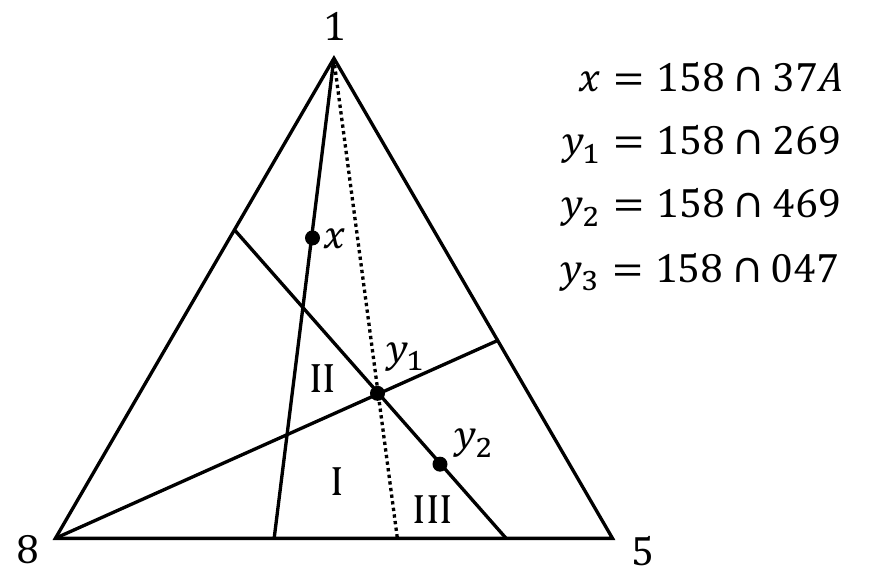}}
\caption{\label{fig_R4} Proof of Lemma~\ref{lem_4Da}. Point $y_3$ must lie in one of the regions I, II, III; more specifically, in region I.}
\end{figure}

Suppose for a contradiction that $P\eqdef (0, 1, \ldots, 9,A,B)$ is orientation-homogeneous and satisfies $2469(7:x)$. Let us look at the relative position of several points inside the triangle $T\eqdef 158$. The hyperplane $H_1 \eqdef 2469$ intersects $T$ along the line passing through $y_1 \eqdef 158 \cap 269$ and $y_2\eqdef 158\cap 469$. The assumption $H_1(7:x)$ is equivalent to $H_1(15x:8)$; hence, within $T$,
\begin{equation}\label{eq_y1y2}
y_1y_2(15x:8).
\end{equation}
Now consider the hyperplane $H_2\eqdef 137A$. It intersects $T$ along the line $1x$. Now, by the movement lemma, the predicate $\Pi[37A({}^{269}_{\,58}\,{}^{469}_{\,58}\,5:8)]$ in $\R^3$ is unavoidable. Hence, by central projection from $1$, in $\R^4$ the predicate
\begin{equation*}
\Pi[137A({}^{269}_{158}\,{}^{469}_{158}\,5:8)] \equiv \Pi[H_2(y_1y_25:8)]
\end{equation*}
is unavoidable. Let us assume $P$ avoids its negation. Hence, within $T$,
\begin{equation}\label{eq_1x}
1x(y_1y_25:8).
\end{equation}
Next, consider the hyperplane $H_3\eqdef 2689$, which intersects $T$ along the line $8y_1$. By the same-side rule we have $H_3(5y_2:1)$. Therefore, within $T$ we have
\begin{equation}\label{eq_8y1}
8y_1(5y_2:1).
\end{equation}
By (\ref{eq_y1y2}), (\ref{eq_1x}), and (\ref{eq_8y1}), the position of $x$, $y_1$, and $y_2$ within $T$ must be as in Figure~\ref{fig_R4}. From the figure we see that, within $T$, $5x(y_1:1)$. Hence, $P$ satisfies the hyperplane-separation statement
\begin{equation*}
s\eqdef 357A(1:{}^{158}_{269}).
\end{equation*}
Now, assume $P$ is homogeneous with respect to the predicate $\Pi[s]$. Then, in particular, $P$ satisfies $2469(0:y_3)$ for $y_3\eqdef 047 \cap 158$. Hence, within $T$ we have
\begin{equation}\label{eq_another_y1y2}
y_1y_2(8y_3:15).
\end{equation}
We will now try to locate $y_3$ more precisely within $T$.

By the same-side rule, we have $137A(8:5{}^{047}_{158})$. Therefore, within $T$, we have
\begin{equation}\label{eq_another_1x}
1x(8:5y_3).
\end{equation}
By (\ref{eq_another_y1y2}) and (\ref{eq_another_1x}) it follows that $y_3$ lies in one of the regions labeled I, II, III in Figure~\ref{fig_R4}.

\setcounter{claim}{0}
\begin{claim}
There exists a line through $y_1$ that separates $y_3$ from $1$ and $8$.
\end{claim}

\begin{proof}
By the movement lemma and central projection from $0$, the predicate $\Pi[0269({}^{047}_{158}5:18)]$ is unavoidable. Let us assume $P$ avoids its negation. Hence, letting $y_4 \eqdef 158\cap 026$, within $T$
we have $y_1y_4(5y_3:18)$.
\end{proof}

This rules out region II for $y_3$.

\begin{claim}
There exists a line through $y_3$ that separates $y_1$ from $1$ and $8$.
\end{claim}

\begin{proof}
By the movement lemma and central projection from $9$, the predicate $\Pi[0479({}^{269}_{158}5:18)]$ is unavoidable. So let us assume $P$ avoids its negation. Let $y_5 \eqdef 479 \cap 158$. Therefore, within $T$ we have $y_3y_5(18:5y_1)$.
\end{proof}

This rules out region III for $y_3$. Hence, $y_3$ lies in region I, which implies that $y_1$ lies inside the triangle $15y_3$.

Now, by the same-side rule, we have $457A(1y_3:{})$. Therefore, again by the same-side rule (since $5$, $1$, and $y_3$ are all on the same side of $457A$, and since $y_1\in 15y_3$),
\begin{equation*}
457A(1y_1:{}).
\end{equation*}
This, however, is an instance of $\neg \Pi[s]$. Contradiction.
\end{proof}

\begin{lemma}\label{lem_4Db}
The predicate $\Pi[248A(6{}^{159}_{37B}:{})]$ is unavoidable.
\end{lemma}

\begin{proof}
Rewrite the predicate as $\Pi[248A(5{}^{159}_{37B}:{})]$ and then as $\Pi[2479(5x:{})]$ where $x\eqdef 158\cap 36A$.

As before, let $P\eqdef (1,\ldots,9,A)$ be orientation-homogeneous, and let us look at the relative position of several points within the triangle $T\eqdef 158$. Let $z_1\eqdef 158\cap 479$. By the same side rule, we have $356A(8z_1:1)$. Hence, within $T$ we have $5x(8z_1:1)$. Furthermore, by the movement lemma and central projection from $1$, we have $136A(8{}^{158}_{479}:5)$. Therefore, within $T$ we have $1x(8z_1:5)$. See Figure~\ref{fig_R4b}.

\begin{figure}
\centerline{\includegraphics{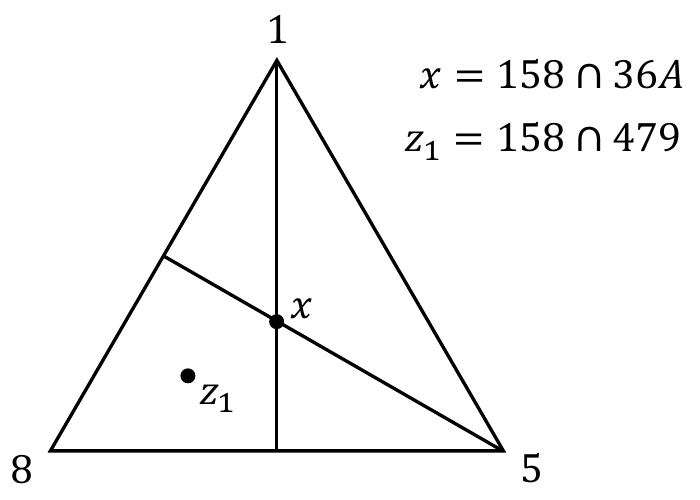}}
\caption{\label{fig_R4b}Proof of Lemma~\ref{lem_4Db}. As we can see, no line through $z_1$ can separate $x$ from $1$ and $5$.}
\end{figure}

Now, suppose for a contradiction that $2479(x:15)$. Then, within $T$ we would have $z_1z_2(x:15)$ for$z_2\eqdef 158\cap 279$; in other words, $x$ would be separated from $1$ and $5$ by a line through $z_1$. But, as we see from Figure~\ref{fig_R4b}, this is impossible.
\end{proof}

\section{The stretched grid and stair-convexity}\label{sec_stair}

In this section we recall the definition of the \emph{stretched diagonal}, and we prove that it is homogeneous with respect to all Tverberg-type predicates,\footnote{The stretched diagonal can be made homogeneous with respect to \emph{any} finite set of semialgebraic predicates, by simply stretching it enough.} and moreover, that the Tverberg types that occur in it are precisely the colorful ones. This constitutes the proof of Theorem~\ref{thm_stretched-diag}.

The stretched diagonal is a subset of a more general construction called the \emph{stretched grid}. The stretched grid yields our result regarding Sierksma's conjecture mentioned in the Introduction.
Hence, we start by describing the stretched grid, and then we go on to the stretched diagonal.
The \emph{stretched grid}, previously introduced in \cite{BMN_epsilonnets,BN_centerlines,N_thesis}, is an axis-parallel grid of points
where, in each direction $i$, $2\le i \le d$, the spacing
between consecutive ``layers'' increases rapidly, and
furthermore, the rate of increase for direction $i$ is much
larger than that for direction $i-1$. To simplify calculations,
we also make the coordinates increase rapidly in the first
direction.

The definition is as follows: Given $n$, the desired
number of points, let $m \eqdef n^{1/d}$ be the side of the grid
(assume for simplicity that this quantity is an integer), and
let
\begin{equation}\label{eq_Gs}
\Gs \eqdef \bigl\{(K_1^{a_1}, K_2^{a_2}, \ldots, K_d^{a_d}) : a_i
\in \{0,\ldots, m-1\} \text{ for all $1\le i\le d$} \bigr\},
\end{equation}
for some appropriately chosen constants $1<K_1 \ll K_2 \ll K_3
\ll \cdots \ll K_d$. Each constant $K_i$ must be chosen
appropriately large in terms of $K_{i-1}$ and in terms of $m$.
Specifically:
\begin{equation}\label{eq_Ki}
K_1 = 2; \qquad K_i \ge 2d^2K_{i-1}^m \quad \text{for $2\le i\le m$.}
\end{equation}

A \emph{vertical projection onto $\R^{d-1}$} is obtained by removing the last coordinate.
We refer to the $d$-th coordinate as the
``height'', so we call a hyperplane in $\R^d$ \emph{horizontal} if
all its points have the same last coordinate; and we call a line in
$\R^d$ \emph{vertical} if its vertical projection is a single point. The \emph{$i$-th
horizontal layer of $\Gs$} is the subset of $\Gs$ obtained by
letting $a_d = i$ in (\ref{eq_Gs}).

The following lemma
provides the motivation for the stretched grid:

\begin{lemma}\label{lemma_property_Gs}
Let $a \in \Gs$ be a point at horizontal layer $0$, and let
$b\in \Gs$ be a point at horizontal layer $i$. Let $c$ be the
point of intersection between segment $ab$ and the horizontal
hyperplane containing layer $i-1$. Then $|c_j - a_j| \le 1/d^2$ for
every $1\le j\le d-1$.
\end{lemma}

Lemma~\ref{lemma_property_Gs} follows from a simple calculation
(we chose the constants $K_i$ in (\ref{eq_Ki}) large enough to
make this and later calculations work out).
The grid $\Gs$ is hard to visualize, so we apply to it a
logarithmic mapping $\pi$ that converts $\Gs$ into the uniform
grid in the unit cube.
Formally, let $\B \eqdef [1,K_1^{m-1}] \times \cdots \times [1,
K_d^{m-1}]$ be the bounding box of the stretched grid, let
$[0,1]^d$ be the unit cube in $\R^d$, and define the mapping
$\pi \colon \B \to [0,1]^d$ by
\begin{equation*}
\pi(x) \eqdef \left( \frac{\log_{K_1} x_1}{m-1}, \ldots, \frac{\log_{K_d} x_d}
{m-1} \right).
\end{equation*}

Then, it is clear that $\pi(\Gs)$ is the uniform grid in
$[0,1]^d$.

\begin{figure}
\centerline{\includegraphics{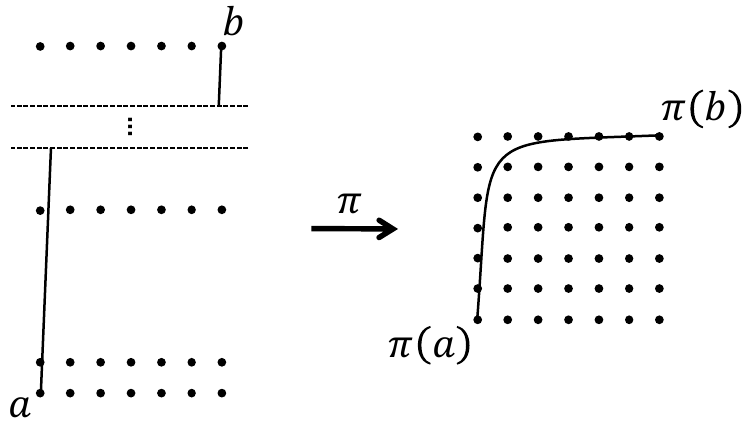}}
\caption{\label{fig_Gs_pi}The stretched grid and the mapping $\pi$ in the plane. The stretched grid is too tall to be drawn entirely, so an intermediate portion of it has been omitted. A line segment connecting two points is also shown, as well as its image under $\pi$. (The first coordinate of the stretched grid does not increase geometrically in this picture.)}
\end{figure}

Lemma~\ref{lemma_property_Gs} implies that the map $\pi$
transforms straight-line segments into curves composed of
almost-straight axis-parallel parts: Let $s$ be a straight-line
segment connecting two points of $\Gs$. Then $\pi(s)$ ascends
almost vertically from the lower endpoint, almost reaching the
height of the higher endpoint, before moving significantly in
any other direction; from there, it proceeds by induction. See Figure~\ref{fig_Gs_pi}.

This observation motivates the notions of
\emph{stair-convexity}, which describe, in a sense, the limit
behavior of $\pi$ as $m\to\infty$.

\subsection{Stair-convexity}

Given a pair of points $a, b\in \R^d$, the \emph{stair-path}
$\sigma(a,b)$ between them is a polygonal path connecting $a$
and $b$ and consisting of at most $d$ closed line segments,
each parallel to one of the coordinate axes. The definition
goes by induction on $d$; for $d=1$, $\sigma(a,b)$ is simply
the segment $ab$. For $d\ge 2$, after possibly interchanging
$a$ and $b$, let us assume $a_d\le b_d$. We set $a' \eqdef (a_1,
\ldots, a_{d-1}, b_d)$, and we let $\sigma(a,b)$ be the union
of the segment $aa'$ and the stair-path $\sigma(a', b)$; for
the latter we use the recursive definition, ignoring the common
last coordinate of $a'$ and $b$.
Note that, if $c$ and $d$ are points along $\sigma(a, b)$, then
$\sigma(c, d)$ coincides with the portion of $\sigma(a, b)$
that lies between $c$ and $d$.

We call a set $S\subseteq \R^d$ \emph{stair-convex} if
for every $a, b\in S$ we have $\sigma(a, b) \subseteq S$.
For a real number $y$ let
$h(y)$ denote the ``horizontal'' hyperplane $\{x\in\R^d: x_d=y\}$.
For a horizontal hyperplane $h\eqdef h(y)$, let $h^+\eqdef\{x\in\R^d: x_d\ge
y\}$ be the upper closed half-space bounded by $h$, and let $h^-$ be
the lower closed half-space. For a set $S\subseteq\R^d$ let $S(y)\eqdef
S\cap h(y)$ be the horizontal slice of~$S$ at height $y$.

For a point $x\eqdef (x_1,\ldots,x_d)\in \R^d$, let $\overline x \eqdef (x_1,
\ldots, x_{d-1})$ be the projection of $x$ into $\R^{d-1}$, and
define $\overline S$ for $S \subset \R^d$ similarly. For a point
$x\in \R^{d-1}$ and a real number $x_d$, let $x \times x_d \eqdef (x_1,
\ldots, x_{d-1}, x_d)$, with a slight abuse of notation.

\begin{lemma}[\cite{BMN_epsilonnets}]
A set $S\subset \R^d$ is
stair-convex if and only if the following two conditions hold:
\begin{enumerate}
\item Every horizontal slice $\overline{S(y)}$ is stair-convex.
\item For
every $y_1 \le y_2 \le y_3$ such that $S(y_3) \neq\emptyset$ we
have $\overline{S(y_1)} \subseteq \overline{S(y_2)}$ (meaning, the horizontal slice
can only grow with increasing height, except that it can end by
disappearing abruptly).
\end{enumerate}
\end{lemma}

Since the intersection of stair-convex sets is obviously
stair-convex, we can define the \emph{stair-convex hull}
$\sconv(S)$ of a set $S\subseteq \R^d$ as the intersection of all
stair-convex sets containing~$S$.

\begin{lemma}[\cite{BMN_epsilonnets}]\label{lemma_sconv}
The stair-convex hull can be characterized by induction on $d$ in the following way: Let $S\subseteq \R^d$, and let $T \eqdef \sconv(S)$. Then for every $y\in \R$, if $h(y)^+ \cap S = \emptyset$, then $T(y) = \emptyset$; otherwise, we have $\overline{T(y)} = \sconv\bigl(\overline{S\cap h^-}\bigr)$. (In other words, the horizontal slice of $T$ at height $y$ is obtained inductively by taking the stair-convex hull in dimension $d-1$ of all points not above height $y$---unless $h(y)$ is strictly above all of $S$, in which case the slice is empty.)
\end{lemma}

\begin{corollary}[Axis-parallel closedness]\label{cor_sconv_closed}
Let $S\subset \R^d$ be a finite point set, let $x\in \sconv(S)$, and let $1\le i\le d$. Let $p$ be the point of $S$ with largest $i$-coordinate that satisfies $p_i\le x_i$, and let $q$ be the point of $S$ with smallest $i$-th coordinate that satisfies $q_i\ge x_i$. Then, if we replace the $i$-th coordinate of $x$ by any real number $p_i\le t\le q_i$, the new point will still belong to $\sconv(S)$.
\end{corollary} 

Let $a\in \R^d$ be a fixed point, and let $b\in \R^d$ be
another point. We say that $b$ has \emph{type $0$ with respect
to $a$} if $b_i\le a_i$ for every $1\le i\le d$. For $1\le j\le
d$ we say that $b$ has \emph{type $j$ with respect to $a$} if
$b_j\ge a_j$ but $b_i\le a_i$ for every $i$ satisfying $j+1\le
i\le d$. (It might happen that $b$ has more than one type with
respect to $a$, but only if some of the above inequalities are
equalities.)

The following lemma is the stair-convex analogue of Carath\'eodory's theorem: 

\begin{lemma}[\cite{BMN_epsilonnets}]\label{lemma_Carath_stair}
Let $S \subseteq \R^d$ be a point set, and let $x\in\R^d$ be a
point. Then $x\in\sconv(S)$ if and only if $S$ contains a point
of type $j$ with respect to $x$ for every $j=0,1,\ldots,d$.
\end{lemma}

The following is a simple claim on regular convexity and its stair-convex analogue:

\begin{lemma}[\cite{N_thesis}]\label{lemma_walk}
\ 

\begin{enumerate}
\item
Let $p\in \R^d$ be a point contained in $\conv(Q)$ for some $Q\subseteq R^d$. Then there exists a $k\le d+1$ and there exist points $q_1,\ldots, q_k\in Q$ and $r_1, \ldots, r_k\in \R^d$ such that $r_1=q_1$, $r_k=p$, and for every $2\le i\le k$ the point $r_i$ lies in the segment $r_{i-1}q_i$. (In other words, we can get to $p$ by starting at $q_1$ and ``walking'' towards $q_2, q_3, \ldots, q_k$ in succession.)

\item
Let $p\in \R^d$ be a point contained in $\sconv(Q)$ for some $Q\subseteq R^d$. Then there exists a $k\le d+1$ and there exist points $q_1,\ldots, q_k\in Q$ and $r_1, \ldots, r_k\in \R^d$ such that $r_1=q_1$, $r_k=p$, and for every $2\le i\le k$ we have $r_i\in\sigma(r_{i-1},q_i)$.
\end{enumerate}
\end{lemma}

We say that a point $p$ \emph{shares the $i$'th coordinate} with a set $S$ if
$p_i=s_i$ for some~$s\in S$.
\begin{lemma}[\cite{BMN_epsilonnets}]\label{lemma_share_coords}
Let $S$ be a $k$-point set in $\R^d$ for some $k\le d+1$, and let
$p$ be a point in $\sconv(S)$. Then $p$ shares at least $d+1-k$
coordinates with $S$.
\end{lemma}

\begin{definition}\label{def_stgenpos}
A set $S\subseteq\R^d$ is said to be in \emph{stair-general position} if $p_i\neq q_i$ for every two distinct points $p,q\in S$ and every $1\le i\le d$.
\end{definition}

The \emph{multipartite Kirchberger theorem}~\cite{Arocha2009,por_thesis} states that if $P_1, \ldots, P_r \subset \R^d$ are $r$ point sets of total size $\sum |P_i| \ge T(d,r)$, such that their convex hulls intersect at a common point $x$, then there exist subsets $P'_1\subseteq P_1$, $\ldots$, $P'_r \subseteq P_r$ of total size $\sum |P'_i| = T(d,r)$, whose convex hulls still intersect at a common point (not necessarily $x$).\footnote{Proof sketch: Suppose the total number of points is larger than $T(d,r)$ and that the points within each $P_i$ are affinely independent. Then the affine hulls of the $P_i$'s intersect at a flat $f$ of dimension at least $1$. Starting at $x$, let us move within $f$  in a straight line, until we first hit the boundary of some $\conv P_i$. This allows us to remove one point from that $P_i$.}

The following lemma is the stair-convex analogue of the multipartite Kirchberger theorem. It generalizes Lemma~5.4 of~\cite{BMN_epsilonnets} from $2$ parts to $r$:

\begin{lemma}\label{lemma_st_multiK}
Let $P\subset \R^d$ be an $n$-point set in stair-general position, and let $P_1, \ldots, P_r$ be a partition of $P$ into $r$ parts. Let $X = \sconv(P_1) \cap \cdots \cap \sconv(P_r)$. Then:
\begin{enumerate}
\item[\rm(a)] If $n<T(d,r)$ then $X=\emptyset$.

\item[\rm(b)] If $n=T(d,r)$ and $X\neq \emptyset$, then $X$ contains a single point. Furthermore, each of the $r$ highest points of $P$ belongs to a different part $P_i$.

\item[\rm(c)] If $n\ge T(d,r)$ and $X \neq \emptyset$, then there exist subsets $Q_i \subseteq P_i$ of total size $\sum |Q_i| = T(d,r)$ such that $\bigcap \sconv(Q_i)\neq \emptyset$.
\end{enumerate}
\end{lemma}

\begin{proof}
Suppose there exists a point $x\in X$.
By Lemma~\ref{lemma_share_coords}, $x$ shares at least $c\eqdef\sum (d+1-|P_i|) = d + (T(d,r)-n)$ coordinates with the points of $P$. Hence, $n<T(d,r)$ would imply $c\ge d+1$, meaning, $x$ shares the same coordinate with two different points of $P$, a contradiction. This proves part (a).

Now suppose $n=T(d,r)$. Hence, $x$ shares all $d$ coordinates with the points of $P$, and the same is true of every other point of $X$. Hence, if there were another point $y\in X$, then we would have $\sigma(x,y)\subseteq X$, which leads to a contradiction since $\sigma(x,y)$ contains infinitely many points.

Now let $\{p_1, \ldots, p_r\}$ be the $r$ highest points of $P$. Suppose for a contradiction that two of them, say $p_1$ and $p_2$, belong to the same part $P_i$, and hence none of them belong to $P_j$ for some $j\neq i$. Then $x$ is not higher than the highest point of $P_j$, which is lower than $p_1$ and $p_2$. Therefore, we could remove one of them, say $p_1$, from $P_i$, and by Lemma~\ref{lemma_sconv} its stair-convex hull would still contain $x$. This would contradict part (a). Hence, we have proven part (b).

We now prove part (c) by induction on $d$. Suppose $n\ge T(d,r)$. If $d=1$ then each $\sconv(P_i)$ is an interval on the real line. Let $y$ be the rightmost point of $X$. Then $y\in P_i$ for some $i$, so we can take one point from $P_i$ and the two extremal points of every other $P_j$, for a total of $T(1,r)=2r-1$ points.

Now suppose $d\ge 2$. Let $h\eqdef h(x_d)$ be the horizontal hyperplane containing $x$, and let $P_i^- \eqdef P_i \cap h^-$ for each $i$. By Lemma~\ref{lemma_sconv}, we have $\overline x \in \sconv\bigl(\overline{P_i^-}\bigr)$ for each $i$. Hence, by induction, we can choose subsets $Q_i^- \subseteq P_i^-$, of total size $T(d-1, r)$, such that $\bigcap \sconv\bigl(\overline{Q_i^-}\bigr)$ is not empty. Let $y$ be a point in this intersection (we do not necessarily have $y=\overline x$).

Let $q$ be the highest point of $\bigcup Q_i^-$, and say $q\in Q_1^-$ for simplicity. Let $z\eqdef y\times q_d$, so by Lemma~\ref{lemma_sconv}, we have $z \in \sconv(Q_1^-)$. For each $2\le i\le r$, let $Q_i \eqdef Q_i^- \cup \{p_i\}$, where $p_i$ is the highest point of $P_i$. Since $p_i$ is not lower than $x$ and $x$ is not lower than $z$, it follows again by Lemma~\ref{lemma_sconv} that $z \in\sconv Q_i$. Hence, the subsets $Q_1^-$, $Q_2$, $Q_3$, $\ldots$, $Q_r$ are the desired subsets of $P_1, \ldots, P_r$, since their total size is $T(d-1,r)+(r-1)=T(d,r)$.
\end{proof}

\begin{definition}
If $P\subset \R^d$ has size $|P| = T(d,r)$ and is in stair-general position, then a \emph{stair-Tverberg partition} of $P$ is one that satisfies $\bigcap_{i=1}^r \sconv(P_i) \neq \emptyset$. The unique point in this intersection is called the \emph{stair-Tverberg point} of this partition.
\end{definition}

It turns out that in stair-convex geometry Sierskma's conjecture is true, and furthermore, there are exactly $(r-1)!^d$ Tverberg partitions, and they have the same Tverberg point:

\begin{lemma}\label{lemma_stTv}
Let $P\subseteq \R^d$ be a point set of size $|P| = T(d,r)$ in stair-general position. Let $p_1, p_2, \ldots, p_{T(d,r)}$ be the points of $P$ listed by decreasing last coordinate. Let $Q \eqdef P\setminus \{p_1, \ldots, p_{r-1}\}$. Then:
\begin{itemize}
\item[\rm(a)] Each stair-Tverberg partition of $P$ is obtained inductively as follows: Let $Q_1, \ldots, Q_r$ be a partition of $Q$ such that $\overline{Q_1}$, $\ldots$, $\overline{Q_r}$ is a stair-Tverberg partition of $\overline{Q}$. Then arbitrarily assign the points $p_1, \ldots, p_{r-1}$ one to each of the $r-1$ parts that \emph{do not} contain $p_r$.
\item[\rm(b)] The stair-Tverberg point of all the above partitions is $x\eqdef y\times p_{rd}$ where $y$ is the stair-Tverberg point of $\overline Q$.
\end{itemize}
\end{lemma}

\begin{proof}
First, let $P_1, \ldots, P_r$ be a stair-Tverberg partition of $P$, and let $x$ be the corresponding stair-Tverberg point. By Lemma~\ref{lemma_st_multiK}, each of the points $p_1, \ldots, p_r$ belongs to a different part $P_i$. Say for simplicity that $p_i\in P_i$ for each $i$.

As pointed out in the proof of Lemma~\ref{lemma_st_multiK}, $x$ shares each of its $d$ coordinates with some point of $P$. Therefore, it must be that $x_d = p_{rd}$. Let $Q_i \eqdef P_i\setminus\{p_i\}$ for $1\le i\le r-1$, and $Q_r \eqdef P_r$. Hence, $Q_1, \ldots, Q_r$ is a partition of $Q$. Further, by Lemma~\ref{lemma_sconv}, we have $\overline x \in \sconv(\overline Q_i)$ for all $i$, as desired.

Now let $Q_1, \ldots, Q_r$ be a partition of $Q$ such that $\overline{Q_1},\ldots, \overline{Q_r}$ is a stair-Tverberg partition of $\overline Q$ with stair-Tverberg point $y$. Say $p_r\in Q_r$ for simplicity. Let $x\eqdef y\times p_{rd}$. Then $x\in Q_r$ by Lemma~\ref{lemma_sconv}. Arbitrarily assign the points $p_1, \ldots, p_{r-1}$ one to each of the sets $Q_1, \ldots, Q_{r-1}$, obtaining sets $P_1, \ldots, P_{r-1}$. Then, again by Lemma~\ref{lemma_sconv}, $x\in P_i$ for each $1\le i\le r-1$. Hence, $P_1, \ldots, P_{r-1}, Q_r$ is a stair-Tverberg partition of $P$ and $x$ is its stair-Tverberg point.
\end{proof}

\begin{corollary}\label{cor_stair_sierksma}
Let $P$ be a $T(d,r)$-point set in $\R^d$ in stair-general position. Then, $P$ has exactly $(r-1)!^d$ stair-Tverberg partitions.
\end{corollary}

\begin{proof}
By induction on $d$. The case $d=1$ is straightforward. For $d\ge 2$, by Lemma~\ref{lemma_stTv}, the number of stair-Tverberg partitions of $P$ equals $(r-1)!$ times the number of stair-Tverberg partitions of $\overline Q$ for the $Q$ mentioned in the lemma.
\end{proof}

\begin{remark}
By Corollary~\ref{cor_sconv_closed}, if $P$ is in stair-degenerate position then the number of partitions of $P$ with intersecting stair-convex hulls can only increase.
\end{remark}

\subsection{A transference lemma}

Let $p, q\in\B$ be two points inside the bounding box of the stretched grid (not necessarily grid points). For $1\le i\le d$, we say that the \emph{stretched distance between $p$ and $q$ in direction $i$} is $c$ if $p_i = K_i^c q_i$ or $q_i = K_i^c p_i$, or, in other words, if $|\pi(q)_i-\pi(p)_i|=c/(m-1)$.

If the stretched distance between $p$ and $q$ in direction $i$ is at most $c$, then we say that $p$ and $q$ are \emph{$c$-close in direction $i$}. If this distance is at least $c$ then we say that $p$ and $q$ are \emph{$c$-far apart in direction $i$}. If $p$ and $q$ are $c$-close in every direction $1,\ldots,d$ then we say that they are \emph{$c$-close}. If they are $c$-far apart in every direction $1,\ldots,d$ then we say that they are \emph{$c$-far apart}.

\begin{lemma}[Transference lemma]\label{lem_transference}
Let $P\subset \B$ be a finite point set such that every two points of $P$ are $(2d+3)$-far apart. Let $P_1, \ldots, P_r$ be a partition of $P$ into $r$ parts. Then $\bigcap \conv(P_i)\neq\emptyset$ if and only if $\bigcap \sconv(P_i)\neq \emptyset$.
\end{lemma}

The lemma is intuitively obvious, given that stair-convexity is the limit behavior of regular convexity in the stretched grid under $\pi$. The reason we need the points of $P$ to be far apart enough from each other is to avoid the ``rounded'' parts of the $\pi(\conv(P_i))$'s---the parts in which the correspondence between convexity and stair-convexity breaks down.

Unfortunately, the proof of the lemma is quite tedious. We relegate it to Appendix~
\ref{app_transference}.

\subsection{Our results on the stretched grid}

This is our result regarding Sierksma's conjecture:

\begin{theorem}\label{thm_our_sierksma}
Let $P$ be any set of $T(d,r)$ points in $\B$ such that every two points of $P$ are $(2d+3)$-far apart. Then $P$ has exactly $(r-1)!^d$ Tverberg partitions.
\end{theorem}

In particular, a randomly chosen set of $T(d,r)$ points from the stretched grid will satisfy the condition of the theorem with probability tending to $1$ as $n\to\infty$.

\begin{proof}[Proof of Theorem~\ref{thm_our_sierksma}]
By the transference lemma and Corollary~\ref{cor_stair_sierksma}.
\end{proof}

\begin{definition}
The \emph{stretched diagonal} is the sequence of points obtained by taking $a_1 = a_2 = \cdots = a_d = (2d+3)j$ for $j=0, 1, 2, \ldots, (m-1)/(2d+3)$ in (\ref{eq_Gs}).
\end{definition}

\begin{lemma}
The stretched diagonal is homogeneous with respect to all Tverberg partitions; moreover, the Tverberg partitions that occur in the stretched diagonal are exactly the colorful ones.
\end{lemma}

\begin{proof}
By the transference lemma and Lemma~\ref{lemma_stTv}, since the stretched diagonal is monotonic with respect to all coordinates simultaneously.
\end{proof}

\bibliographystyle{alpha}
\bibliography{unavoidable_Tv}

\appendix

\section{Proof of the transference lemma (Lemma~\ref{lem_transference})}\label{app_transference}

We start with some simple claims about axis-parallel boxes.

Recall that $[d] \eqdef \{1,\ldots, d\}$. Let $C\subset \R^d$ be an axis-parallel box $C \eqdef \{x\in\R^d\mid a_i \le x_i \le b_i \text{ for $i\in[d]$}\}$. For each $i\in[d]$, call the facet of $C$ that satisfies $x_i=a_i$ \emph{facet $-i$}, and the facet of $C$ that satisfies $x_i=b_i$ \emph{facet $i$}.

\begin{lemma}\label{lemma_I_not_too_large}
A $k$-flat that intersects the interior of $C$ must intersect a facet $\pm i$ for at least $k$ distinct indices $i\in[d]$.
\end{lemma}

\begin{proof}
If $k=0$ there is nothing to prove. If $k=1$ the claim is simple, since a line must enter $C$ through some facet.

For $k\ge 2$ we proceed by induction on $d$. Let $f$ be a $k$-flat, $k\ge 2$, that contains point $p$ the interior of $C$. Let $I$ be the set of indices $i\in[d]$ such that $f$ \emph{does not} intersect the facets $\pm i$. We want to show that $|I| \le d-k$. As pointed out above, we certainly have $I\neq [d]$. Hence, let $i\notin I$. Let $h$ be the hyperplane given by $x_i=p_i$. Then $C'\eqdef C\cap h$ is a $(d-1)$-dimensional axis-parallel box, whose $d-1$ pairs of facets can be labeled $\pm j$ for $j\in[d]\setminus\{i\}$ in the natural way. Further, $f'\eqdef f\cap h$ is a $(k-1)$-dimensional flat, which does not intersect any of the facets $\pm j$, $j\in I$. Hence, by induction on $d$ we have $|I| \le (d-1)-(k-1) = d-k$.
\end{proof}

\begin{definition}
A $k$-flat $f$ ($0\le k\le d$) is said to be \emph{$I$-oriented with respect to $C$}, for $I\subseteq [d]$ of size $|I|=d-k$, if $f$ intersects the interior of $C$ but does not intersect any of the facets $\pm i$, $i\in I$.
\end{definition}

\begin{lemma}
If the $k$-flat $f$ is $I$-oriented with respect to $C$, then $f$ intersects the interior of every facet $\pm i$, $i\notin I$.
\end{lemma}

\begin{proof}
If $k=0$ there is nothing to prove. If $k=1$ the claim is simple, since $f$ is a line, which must enter and exit $C$ through two distinct facets, and there are only two facets left.

For $k\ge 2$ we proceed by induction on $d$. Let $p$ be a point of $f$ lying in the interior of $C$. Pick an index $i\in[d]\setminus I$, and let $h$ be the hyperplane given by $x_i=p_i$. As before, let $C'\eqdef C\cap h$, and label its facets $\pm j$ for $j\in[d]\setminus\{i\}$ in the natural way. Then $f'\eqdef f\cap h$ intersects the interior of $C'$, but does not intersect any of its facets $\pm j$, $j\in I$. The dimension of $f'$ is either $k-1$ or $k$, but it cannot be $k$ by Lemma~\ref{lemma_I_not_too_large}, so it is $k-1$. Hence, $f'$ is $I$-oriented with respect to $C'$. Hence, by induction, $f'$ intersects the interior of all the facets of $C'$ labeled $\pm j$ for $j\notin I\cup \{i\}$. Therefore, $f$ intersects the interior of the equally-named facets of $C$.

To prove that $f$ intersects the interior of the facets $\pm i$, repeat the above argument with a different index $j\in [d]\setminus I$. Such a $j\neq i$ is guaranteed to exist since $k\ge 2$.
\end{proof}

\begin{lemma}\label{lemma_Ioriented_intersect}
Let $I_1, \ldots, I_m$ be $m$ pairwise-disjoint subsets of $[d]$ whose union equals $[d]$. Let $f_1,\ldots, f_m$ be flats that are $I_1$-, $\ldots$, $I_m$- oriented with respect to $C$, respectively. Then $f_1\cap \cdots \cap f_m$ contains a single point, which lies in the interior of $C$.
\end{lemma}

\begin{proof}
Without loss of generality assume there is no $j$ for which $I_j = \emptyset$, since that implies $f_j=\R^d$.

We first prove the claim for the special case where each $I_j$ has size $1$ (so the flats $f_j$ are hyperplanes and $m=d$). In this case we proceed by induction on $d$. Say for simplicity that $I_i=\{i\}$ for each $1\le i\le d$.

For each $a_1<z<b_1$ let $h(z)$ be the hyperplane satisfying $x_1 = z$. Consider the $(d-1)$-dimensional axis-parallel box $C'(z) = C \cap h(z)$. Let $i\ge 2$. Since $f_i$ intersects facets $\pm 1$ of $C$, by convexity it intersects $C'(z)$. Hence $f'_i(z) \eqdef f_i\cap C'(z)$ is a $(d-2)$-flat (a hyperplane within $h(z)$) that intersects the interior of $C'(z)$ but avoids its facets $\pm i$ (since they are contained in the equally-named facets of $C$). Hence, $f'_i(z)$ is $\{i\}$-oriented with respect to $C'(z)$.

Therefore, by induction on $d$, the hyperplanes $f'_2(z), \ldots, f'_d(z)$ intersect at a point $p(z)$ in the interior of $C'(z)$. The points $p(z)$, $a_1<z<b_1$, form a line segment. This line segment goes from one side of $f_1$ to the other one, so it must intersect $f_1$. This proves the special case.

The general case can be reduced to the above special case by applying the following claim:

\setcounter{claim}{0}
\begin{claim}
Let $f$ be a $k$-flat that is $I$-oriented with respect to $C$. Then $f$ can be written as the intersection of $d-k$ hyperplanes, each of which is $\{i\}$-oriented with respect to $C$ for a different $i\in I$.
\end{claim}

\begin{proof}
Let $e_1, \ldots, e_d$ be the standard unit vectors in $\R^d$.

We first show that that if $g$ is a $J$-oriented $\ell$-flat and $j\in J$, then $g+\R e_j\eqdef \{x + \alpha e_j \mid x \in g, \alpha \in \R\}$ (the extrusion of $g$ in the direction $e_j$) is a $(J\setminus\{j\})$-oriented $(\ell+1)$-flat. Indeed, suppose for a contradiction that $g+\R e_j$ intersects a facet $\pm i$ in $J\setminus\{j\}$. That means that there is some point $y$ in $g$ (but outside of $C$) such that $y + \alpha e_j$ falls on facet $\pm i$. Let $q$ be a point of $g$ in the interior of $C$. Then the segment $qy$, which lies in $g$, intersects one of the facets $\pm j$---contradiction.

Now let us come back to the given $k$-flat $f$. For each $i\in I$, let $h_i$ be the hyperplane obtained by extruding $f$ in all directions $e_j$, $j\in I\setminus\{i\}$. By the above argument, $h_i$ is $\{i\}$-oriented with respect to $C$. Hence, these are the desired $d-k$ hyperplanes.
\end{proof}

This concludes the proof of Lemma~\ref{lemma_Ioriented_intersect}.
\end{proof}

\subsection{Back to the stretched grid}

The following lemma is the main motivation behind the stretched grid. It says that, under $\pi$, straight-line segments become very close to stair-paths.

\begin{lemma}\label{lemma_1close}
Let $a,b\in\B$. Then every point of the line segment $ab$ is $1$-close to some point of the stair-path $\sigma(a,b)$ and vice versa.
\end{lemma}

\begin{proof}[Proof sketch]
Here we give a sketch of the proof. The full proof can be found in \cite{N_thesis}. Say without loss of generality that $a_d\le b_d$. Let $c$ be the lowest point of the segment $ab$ that is $1$-close to $b$ in direction $d$. Let $c'$ be the point directly above $a$ at the same height as $b$. As mentioned above (Lemma~\ref{lemma_property_Gs}), for each $1\le i\le d-1$ we have $|c_i-c'_i|\le 1/d^2$.

Hence, let us split the segment $ab$ into segments $ac$ and $cb$, and let us split the stair-path $\sigma(a,b)$ into the vertical segment $ac'$ and the $(d-1)$-dimensional stair-path $\sigma(c',b)$. Then every point of $ac$ is $1$-close to a point of $ac'$ and vice versa, with plenty of room to spare in the first $d-1$ coordinates. And every point of $cb$ is $1$-close in direction $d$ to a point of $\sigma(c',b)$. For directions $1,\ldots,d-1$ we would like to argue by induction on $d$. The problem is that $\overline c$ and $\overline{c'}$ do not exactly coincide.

The way to solve this problem is to prove by induction a stronger claim: If $a'$ is very close to $a$ and $b'$ is very close to $b$, then every point of the line segment $ab$ is $1$-close to some point of the stair-path $\sigma(a', b')$ and vice versa. We omit the details which are technical and not very hard.
\end{proof}

\begin{corollary}\label{cor_dclose}
Let $P\subseteq \B$ be a point set. Then every point of $\conv(P)$ is $d$-close to a point of $\sconv(P)$ and vice versa.
\end{corollary}

\begin{proof}
By Lemma~\ref{lemma_walk}, applying Lemma~\ref{lemma_1close} $k-1$ times for $k\le d+1$.
\end{proof}

We finally get to the transference lemma:

\begin{proof}[Proof of Lemma~\ref{lem_transference}]
By the multipartite Kirchberger theorem and its stair-convex analogue, it is enough to focus on the case $|P|=T(d,r)$, so let us assume this is the case.

For the first direction, suppose $P_1, \ldots, P_r$ is a stair-Tverberg partition of $P$, and let $p$ be its stair-Tverberg point. For each $i$, let $k_i\eqdef |P_i|$, and let $I_i\subseteq [d]$ be the set of coordinates that $p$ shares with $P_i$. By the proof of Lemma~\ref{lemma_st_multiK}, $|I_i|=d+1-k_i$, and the sets $I_i$ form a partition of $[d]$. Let $f_i$ be the $(k_i-1)$-flat satisfying the equations $x_j=p_j$ for $j\in I_i$ (so $f_i$ is ``tangent'' to $\sconv(P_i)$ at point $p$).

Let $C$ be an axis-parallel box containing $p$ such that each facet of $C$ is at stretched distance $d+1$ from $p$. Hence, each $f_i$ is $I_i$-oriented with respect to $C$.

\setcounter{claim}{0}
\begin{claim}
We have $\sconv(P_i)\cap C = f_i \cap C$ for each $i$.
\end{claim}

\begin{proof}
Suppose for a contradiction that $\sconv(P_i)\cap C$ contains a point $x\notin f_i$. Then $x$ shares with $P_i$ a \emph{different} subset $I'_i\subset [d]$ of $d+1-k_i$ coordinates. Hence, $I'_i$ intersects some $I_j$, $j\neq i$. Say $\ell\in I'_i\cap I_j$. Since both $p$ and $x$ are contained in $C$, this means that a point of $P_i$ is too close to a point of $P_j$ in coordinate $\ell$---contradiction. 

Similarly, suppose for a contradiction that $\sconv(P) \cap C$ is a strict subset of $f_i \cap C$. Then the relative boundary of $\sconv(P_i)\cap f_i$ contains some point $x\in C$. By axis-parallel closedness (Corollary~\ref{cor_sconv_closed}), $x$ shares with $P_i$ an additional coordinate not in $I_i$, yielding a similar contradiction.
\end{proof}

We now examine how $\conv(P_i)$ intersects $C$. We first note that $\conv(P_i)\cap C$ is not empty, by Corollary~\ref{cor_dclose} and by our choice of the size of $C$.

By an argument similar to the one above, we have $\conv(P_i) \cap C = \ahull(P_i) \cap C$, where $\ahull$ denotes the affine hull. Indeed, otherwise $C$ would intersect a facet of $\conv(P_i)$, given by $\conv(P')$ for some strict subset $P'\subsetneq P_i$. Then, by Corollary~\ref{cor_dclose}, $\sconv(P')$ would contain a point $q$ that is $(2d+1)$-close to $p$. This point $q$ shares one more coordinate with $P_i$ than $p$ does, leading to a contradiction as before.

Finally, we claim that $\ahull(P_i)$ is $I_i$-oriented with respect to $C$, just like $f_i$. Indeed, suppose for a contradiction that $\ahull(P_i)$ intersects a facet $j$ for $\pm j\in I_i$, at a point $q$. By Corollary~\ref{cor_dclose}, there exists a point $z\in \sconv(P_i)$ that is $d$-close to $q$; in particular, $z_j\neq p_j$. Let $w$ be the point on facet $j$ of $C$ satisfying $w_\ell=x_\ell$ for all $\ell\neq j$. Then, by axis-parallel closedness of $\sconv(P_i)$ (repeated application of Corollary~\ref{cor_sconv_closed}), we have $w\in \sconv(P_i)$, contradicting the fact that $\sconv(P_i)$ is $I_i$-oriented. See Figure~\ref{fig_in_box}.

\begin{figure}
\centerline{\includegraphics{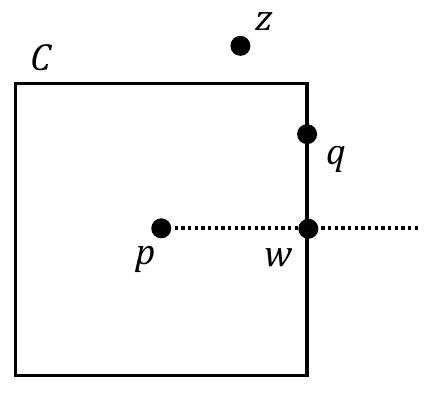}}
\caption{\label{fig_in_box}We have $p\in\sconv(P_i)$. Suppose for a contradiction that there exists a point $q\in\conv(P_i)$ at a facet $j$ of $C$ (here $j=+1$). Then there exists a point $z\in\sconv(P_i)$ not too far away. But then the point $w$, which lies on face $j$ and is aligned with $p$, also belongs to $\sconv(P_i)$.}
\end{figure}

Hence, Lemma~\ref{lemma_Ioriented_intersect} applies, so $\bigcap \conv(P_i)$ contains a point in the interior of $C$, as desired.

For the second direction, suppose that $P_1, \ldots, P_r$ is a Tverberg partition of $P$, and let $p$ be its Tverberg point. Let $C$ be an axis-parallel box containing $p$ such that each facet of $C$ is at stretched distance $d+1$ from $p$. By Corollary~\ref{cor_dclose}, for each $1\le i\le r$, there is a point $y_i\in \sconv(P_i)$ in the interior of $C$. By Lemma~\ref{lemma_share_coords}, $y_i$ shares a set $I_i\subseteq [d]$ of coordinates with $P_i$, where $|I_i| \ge d+1-k_i$. However, since every two points of $P$ are $(2d+3)$-far apart, the sets $I_i$ must be pairwise disjoint. Further, $\sum (d+1-k_i) = d$, so $|I_i| = d+1-k_i$, and the sets $I_i$ form a partition of $[d]$.

As before, for each $i$ we have $\sconv(P_i)\cap C = f_i\cap C$ for some axis-parallel $(k_i-1)$-flat $f_i$ that is $I_i$-oriented with respect to $C$. Hence, Lemma~\ref{lemma_Ioriented_intersect} applies,\footnote{Actually, Lemma~\ref{lemma_Ioriented_intersect} is overkill in this case.} so $\bigcap \sconv(P_i)$ contains a point in the interior of $C$, as desired.
\end{proof}

\section{Mathematica code}\label{app_code}

The following Mathematica code generates all Tverberg types with $d=r=3$ in which the parts have sizes $3,3,3$ and the three triangles pairwise intersect.

\begin{Verbatim}[fontsize=\small]
  s333 = Select[Permutations[{1,1,1,2,2,2,3,3,3}], FirstPosition[#,1][[1]] <
    FirstPosition[#,2][[1]] < FirstPosition[#,3][[1]] &];
  contains[l_,subl_] := Length[LongestCommonSequence[l,subl]] == Length[subl];
  trIntersect[s_,a_,b_] := contains[s, {a,b,a,b,a}] || contains[s, {b,a,b,a,b}];
  Select[s333, trIntersect[#,1,2] && trIntersect[#,1,3] && trIntersect[#,2,3] &]
\end{Verbatim}

In the following code, we define a function \codefont{colorfulTv} that generates all colorful Tverberg types with given parameters $d$ and $r$. We use it to generate the list \codefont{s3334} of all colorful Tverberg types with $d=3$, $r=4$ with parts of sizes $3,3,3,4$. Then, from \codefont{s3334} we generate all the underlying $12$-ary predicates that involve a plane and the point of intersection of three triangles. We encode each such $12$-ary predicate by a sequence on the symbols $\{a,b,c,x\}$, where the symbols $a$, $b$, $c$ represent the three triangles, and $x$ represents the plane. For example, the predicate of Lemma~\ref{lem_highly_symm} is encoded by $abcxabcxabcx$. We make sure each encoding is \emph{canonical}, in the sense that the first $a$ in the sequence appears before the first $b$, which appears before the first $c$.

\begin{Verbatim}[fontsize=\small]
  extensions[r_,s_] := Join[s,#]& /@ Permutations[Complement[Range[r], {s[[-1]]}]];
  extendall[r_,slist_] := Flatten[extensions[r, #] & /@ slist, 1];
  colorfulTv[d_,r_] := Nest[extendall[r, #] &, {Range[r]}, d];
  s3334 = Select[colorfulTv[3, 4], Min[Table[Count[#, i], {i, 4}]] == 3 &];
  makeCanonical[s_,x_] := Module[{abcfirst, rules},
    abcfirst = Sort[{FirstPosition[s, #][[1]], #} & /@ Complement[Range[4], {x}]];
    rules = {x -> "x"}~Join~Table[abcfirst[[i,2]] -> {"a","b","c"}[[i]], {i, 3}];
    s /. rules];
  makePredicates[s_] := Module[{x, xpositions},
    x = Select[Range[4], Count[s, #] == 4 &][[1]];
    xpositions = Flatten[Position[s, x], 1];
    makeCanonical[Delete[s, #], x] & /@ xpositions];
  Union[Flatten[makePredicates /@ s3334, 1]]
\end{Verbatim}

\end{document}